\documentclass[twocolumn,10pt]{IEEEtran}
\usepackage{algorithm,algorithmic,amsbsy,amsmath,amssymb,epsfig,bbm,mathrsfs,fancyhdr,fancyvrb,subfigure,url,color}

\usepackage{bbm}
\usepackage{cite}
\usepackage{longtable,multirow,subfigure}
\usepackage{wrapfig,lipsum,booktabs}
\usepackage{multirow}

\usepackage{soul}

\newcommand{\beq}{\begin{equation}}
\newcommand{\eeq}{\end{equation}}
\newcommand{\beqn}{\begin{eqnarray}}
\newcommand{\eeqn}{\end{eqnarray}}

\newtheorem{theorem}{\textbf{\text{Theorem}}}
\newtheorem{lemma}{\textbf{\text{Lemma}}}

\newtheorem{corollary}{\textbf{\text{Corollary}}}
\newtheorem{assumption}{\textbf{\text{Assumption}}}

\newenvironment{proof}[1][Proof]{\begin{trivlist}
\item[\hskip \labelsep {\bfseries #1}]}{\end{trivlist}}

\newcommand{\eqdef}{\mathrel{\mathop:}=}

\newcommand{\qed}{\nobreak \ifvmode \relax \else
      \ifdim\lastskip<1.5em \hskip-\lastskip
      \hskip1.5em plus0em minus0.5em \fi \nobreak
      \vrule height0.75em width0.5em depth0.25em\fi}
      
\begin{document}

\title{Analytical Modeling of  Mode Selection and Power Control for Underlay D2D Communication in Cellular Networks}

\author{Hesham ElSawy and Ekram Hossain 

\vspace{-2mm}
}

\maketitle

\begin{abstract}

Device-to-device (D2D) communication enables the user equipments (UEs) located in close proximity to bypass the cellular base stations (BSs) and directly connect to each other, and thereby, offload traffic from the cellular infrastructure. D2D communication can improve spatial frequency reuse and energy efficiency in cellular networks. This paper presents a comprehensive and tractable analytical framework for D2D-enabled uplink cellular networks with a  flexible mode selection scheme along with truncated channel inversion power control. Different from the existing mode selection schemes where the decision on mode selection is made based only on the D2D link distance (i.e., distance between two UEs using D2D mode of communication), the proposed mode selection scheme for a UE accounts for both the D2D link distance and cellular link distance (i.e., distance between the UE and the BS). The developed framework is used to analyze and understand how the underlaying D2D communication affects the cellular network performance. Through comprehensive numerical analysis, we investigate the expected performance gains and provide guidelines for selecting the network parameters.


\end{abstract}

\section{Introduction}

Maximizing spectrum utilization via spatial frequency reuse has always been a major technical challenge for cellular network designers. The challenge has become more acute with the increased  population of cellular users and their traffic requirements. Enabling device-to-device (D2D) communication in cellular networks has recently been proposed  as a promising solution to improve the spatial frequency reuse and boost up the throughput of cellular networks~\cite{Err, 3GP1, D2D_sur, D2D_1}. The main idea in D2D-enabled cellular networks is to permit transmitter-receiver pairs coexisting in close proximity to establish  direct peer-to-peer connections between each other. That is, as shown in Fig.~\ref{mode}, if a transmitter has its designated receiver within its transmission range (also called proximity detection region), the transmitter is allowed to bypass the base station (BS) and communicate in the D2D mode (i.e., directly establish a peer-to-peer link with the receiver). D2D communication enables short-range, low-power links to coexist with the cellular links and thereby improve the spatial reuse of the available spectrum, decrease the power consumption in the user equipments (UEs) via decreasing the required transmit power, and improve the total network throughput\cite{D2D_sur, D2D_1, D2D_2, D2D_3, D2D_4, swarm, D2D_st, D2D_SG1, D2D_SG2, and-d2d}. However, D2D communication poses a set of new technical challenges  which include interference management in the network, resource allocation for D2D and cellular links, and adaptive mode selection and power control for the UEs.

\begin{figure}[t]
	\begin{center}
		\includegraphics[trim= 2cm 3.5cm 2cm 4cm, clip, width=2.6 in]{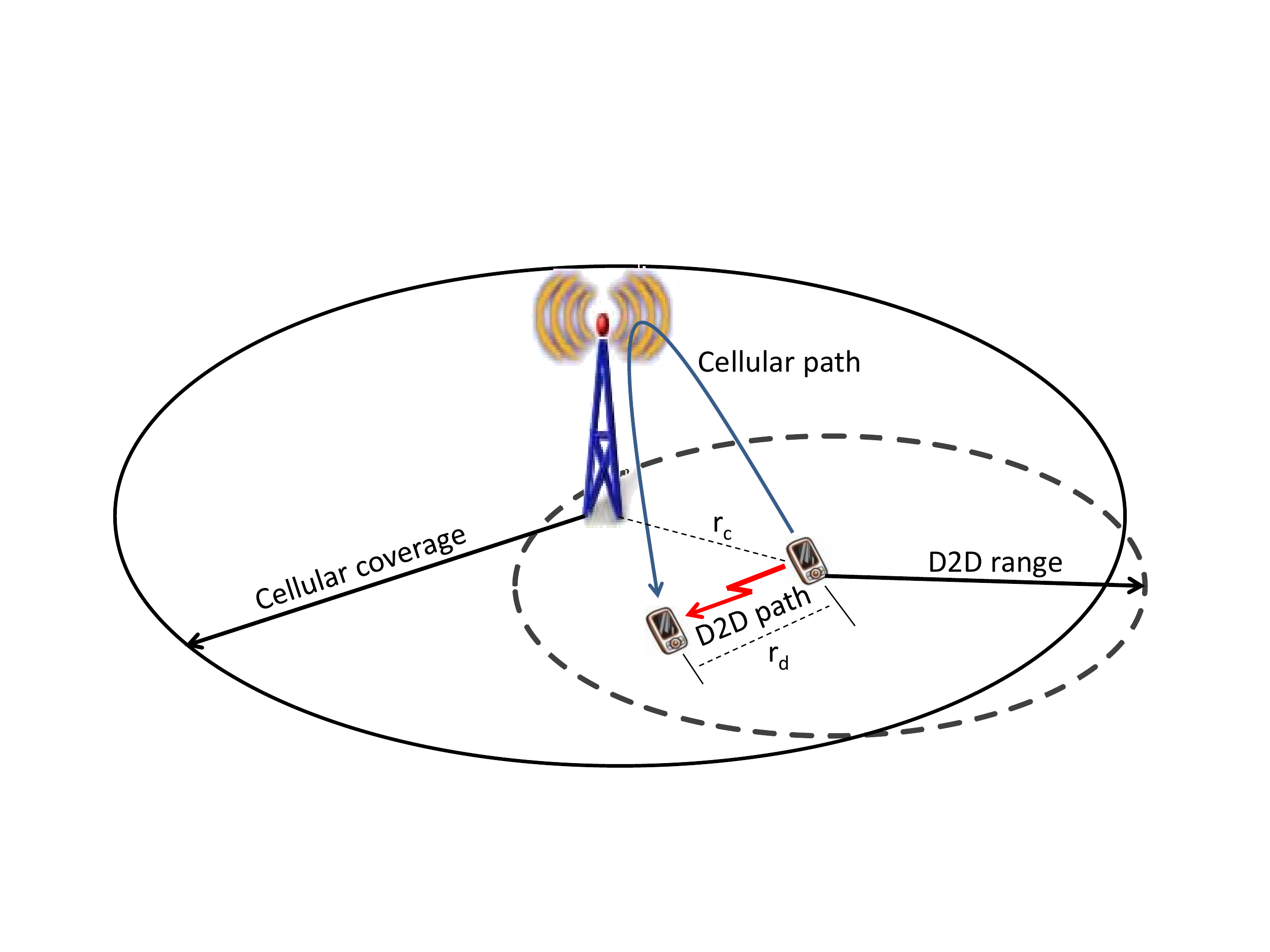}
	\end{center}
	\caption{D2D-enabled cellular network.}
	\label{mode}
\end{figure}
  
There are two main approaches for spectrum assignment between D2D links and cellular links, namely, the disjoint and shared spectrum assignments. It is well-known that rigid and exclusive spectrum assignment results in a significant spectrum underutilization due to the high variability in user population and their activity patterns across the spatial and time domains \cite{our-survey}. Therefore,  the shared spectrum assignment (also called underlay spectrum access) for D2D devices is generally preferred over the disjoint spectrum assignment approach~\cite{D2D_sur}. However, with the underlay spectrum access, cross-mode\footnote{ We use cross-mode interference to refer to the interference between users operating in the D2D mode and users operating in the cellular mode.} interference may degrade the signal-to-interference-plus-noise-ratio (${\rm SINR}$) for the ongoing transmissions and thus limit the network performance. In this paper, we consider the underlay spectrum sharing approach and develop a tractable modeling paradigm to understand and quantify the performance gain for uplink communication in D2D-enabled cellular networks. The analytical framework developed in this paper  uses tools from stochastic geometry \cite{martin-book} and accounts for power control, maximum transmit power of the UEs, and mode selection of UEs in an uplink cellular network.

Stochastic geometry is the only mathematical tool that is capable of modeling large scale wireless networks and capture the effect of network topology on network performance (averaged over all network realizations) and provide general yet simple expressions for several important performance metrics \cite{our-survey}. Stochastic geometry has been used to model, analyze, and design of ad hoc networks, multi-tier cellular networks, as well as hybrid networks \cite{our-survey, martin-book}. For the sake of analytical tractability, we use the Poisson point process (PPP) to model the cellular network topology as well as the spatial distribution of UEs. It has been shown in \cite{our-survey, pp-cellular, valid, trac, trac2} that the PPP leads to tight bound for the cellular network operation which is as accurate as the bound obtained via the idealistic grid based model. It is worth mentioning that our objective is not to develop sophisticated interference mitigation and cancellation techniques for D2D-enabled cellular networks. Instead, our main objective is to develop a tractable analytical framework to quantify the uplink transmission performance of underlay D2D-enabled cellular networks with a biasing-based mode selection scheme along with channel inversion-based power control and identify the performance tradeoffs due to the underlaying D2D communication. 

The major contributions of the paper can be summarized as follows:
\begin{itemize}

\item It proposes a novel  mode selection scheme for UEs in which the mode selection decision accounts for both the D2D link quality and the cellular link quality. More specifically, a biasing-based mode selection scheme with parameter $T_d$ (the bias factor) is proposed in which a potential D2D transmitter chooses the D2D mode if the biased D2D link quality is better than the cellular uplink quality, i.e.,  $T_d L^{(D2D)} \geq L^{(cellular)}$, where $L^{(D2D)}$ is the D2D link gain (or link quality) and $L^{(cellular)}$ is the cellular uplink gain. The proposed mode selection scheme captures the disabled D2D mode of communication (i.e., when $T_d=0$), the enforced D2D communication (i.e., when $T_d = \infty$), and the distance-based mode selection as special cases. 

\item  For the proposed mode selection scheme along with a truncated channel inversion-based power control for UEs, we develop a tractable analytical framework to evaluate the network performance in terms of ${\rm SINR}$ outage probabilities for both cellular and D2D UEs, average transmit power, average link capacity, and average total network capacity.


\item Through extensive numerical analysis, we show that the proposed mode selection scheme outperforms the traditional mode selection scheme based on the D2D link distance only. Also, the results reveal interesting tradeoff in the network performance in terms of average transmit power and ${\rm SINR}$ outage probability (and hence link capacity) when the network parameters such as the power control cutoff threshold and bias factor for mode selection $T_d$ are varied. 

\end{itemize}


The rest of the paper is organized as follows. The related literature is reviewed in Section~\ref{sec:related}. The system model, assumptions, and the mode selection scheme are described in Section~\ref{sys}. Also, the methodology of analysis is outlined in this section. Section~\ref{power} presents the analysis of transmit power for cellular and D2D UEs for channel inversion-based power control. The ${\rm SINR}$ performance is analyzed in~Section~\ref{ana}. Section~\ref{res} presents and discusses the numerical and simulation results. The paper is concluded in Section~\ref{con}. 
A list of the key mathematical notations used in this paper is given in Table~\ref{table}. 

\section{Related Work}
\label{sec:related}

Motivated by the expected gains offered by the underlay D2D communication, research efforts have been invested to analyze and optimize its operation. We can broadly classify the related approaches in the literature into two main categories:  instantaneous analysis approach and statistical analysis approach. In the former approach, a system objective function is formulated based on the the instantaneous system information (e.g., channel gains and link distances), which is assumed to be available. Then, the model is used to derive instantaneous optimal decisions (e.g., power allocation, channel allocation, and mode selection criterion) \cite{D2D_1, D2D_2, D2D_3, D2D_4, swarm}. Note that, the instantaneous optimal decisions should vary with the rapidly varying system parameters. On the other hand, the statistical approach exploits the system's statistical information (e.g., the distributions of the UEs' locations and channel gains), which are stable over a longer  period of time (i.e., w.r.t. the instantaneous approach), to model the system and derive the statistically optimal decisions \cite{D2D_st, D2D_SG1, D2D_SG2, and-d2d}.  Since finding the instantaneous optimal decisions may involve high signaling overhead to exchange the network information as well as high computational complexity, often suboptimal heuristic solutions \cite{D2D_3, D2D_4, swarm} are sought. In contrast to the instantaneous approach, a decision based on the  the statistical approach (for example, based on stochastic geometry analysis) may not be the best solution in a particular point of time, however, it could be optimal over a longer time horizon.

The authors in \cite{D2D_st} exploit the statistical approach to propose a simple power control mechanism for a D2D transmitter to ensure that the ${\rm SINR}$ violation for the cellular users due to cross-mode interference is maintained below a certain threshold. However, the analysis in \cite{D2D_st} is limited to a single cell, a single cellular user, and a single D2D link. In \cite{D2D_SG1}, the authors use the statistical approach to find the maximum intensity of D2D devices that can be accommodated by the uplink cellular network subject to an interference threshold. The authors in \cite{D2D_SG2} also use the statistical approach to find the optimal intensity and transmit power that maximize the achievable transmission capacity for a D2D-enabled uplink two-tier wireless network under outage probability constraints. In \cite{D2D_SG2}, it was assumed that the two cellular tiers operate in disjoint bands and that the D2D links  utilize both the bands. Then, the authors derive the optimal D2D link density and the fixed transmit power that they can use on each band. In both \cite{D2D_SG1} and \cite{D2D_SG2}, power control and mode selection are ignored and it is assumed that the D2D link distances are fixed. Mode selection and power control problems for D2D-enabled uplink cellular networks are considered in \cite{and-d2d}. However, the mode selection decision is based only on the D2D link distance (i.e., it does not consider the distance between the  D2D transmitter and the cellular BS). Furthermore, in \cite{and-d2d}, the maximum power constraint for the UEs is ignored and the coverage area of the tagged cell is approximated by a circle\footnote{It was shown in \cite{our-uplink} that,  under a maximum transmit power constraint for the UEs, for uplink analysis, the circular approximation for  the coverage of the tagged BS is inaccurate.}.  

Different from the exisiting D2D communication models in the literature, our proposed model is based on a flexible mode selection scheme which accounts for  D2D link quality as well as the cellular link quality. The traditional D2D link distance-based mode selection criterion is a special case of the proposed mode selection criterion. The analysis of the model avoids the circular approximation for the BS coverage area and also accounts for power control with a maximum transmit power constraint for the UEs in a multi-cell environment.

\begin{table}
\caption{List of key notations}
\centering
\begin{tabular}{|l |l|}
\hline
Notation & Definition \\
\hline \hline

$\mathbf{\Psi}$ & PPP constituted by the macro BSs  \\ \hline
$\lambda$ & Intensity of macro BSs  \\ \hline
$\mathbf{\Phi}$ & PPP constituted by the UEs  \\ \hline
$\mathcal{U}$ & Intensity of UEs  \\ \hline
$\mathcal{D}$ & Intensity of potential D2D UEs  \\ \hline
$\rho_{min}$ & Receiver sensitivity \\ \hline
$\rho_{o}$ & Power control cutoff threshold \\ \hline
$R_{max}$ & Maximum D2D communication range based on $\rho_{min}$ \\ \hline
$R$ & D2D communication range based on $\rho_o$ \\ \hline
$1-p$ & D2D truncation probability \\ \hline
$h$ & Small-scale fading channel gain \\ \hline
$\theta$ & Required ${\rm SINR}$ threshold \\ \hline
$T_d$ & D2D bias factor  \\ \hline
$\mathcal{P}_d$ & D2D mode selection probability  \\ \hline
$r_c$ & Distance from a generic UE to the nearest BS  \\ \hline
{\multirow{2}{*}{$r_d$}} & Distance from a generic potential D2D UE to its \\
 & receiver  \\ \hline
{\multirow{2}{*}{$\tilde{r}_c$}} & Conditional distance from a generic UE operating in \\ 
$$ & the cellular mode to the nearest BS  \\ \hline
{\multirow{2}{*}{$\tilde{\tilde{r}}_c$}} & Conditional distance from a generic UE operating in   \\
$$ &  cellular mode (for case \#2) to the nearest BS  \\ \hline
{\multirow{2}{*}{$\tilde{r}_d$}} & Conditional distance from a generic D2D UE  \\ 
$$ & operating in the D2D mode to its receiver  \\ \hline
{\multirow{2}{*}{$P_c$}} & Transmit power of a generic UE operating in cellular   \\
$$&    mode  \\ \hline
{\multirow{2}{*}{$P_d$}} & Transmit power of a generic UE operating in D2D \\
$$ &   mode  \\ \hline
$P_u$ & Maximum transmit power of a UE  \\ \hline
${P}_2$ & Transmit power of a generic UE (for case $\# 2$) \\ \hline
{\multirow{2}{*}{$\tilde{P}_4$}} & Conditional transmit power of a potential   \\
$$ &D2D UE operating in cellular mode \\ \hline
$\eta_c$ & Path-loss exponent for cellular link  \\ \hline
$\eta_d$ & Path-loss exponent for D2D link  \\ \hline
$\sigma^2$ & Noise power  \\ \hline
{\multirow{2}{*}{$\chi$}} & Operation mode of a UE   \\ 
           & ($\chi \in$ \{cellular mode, D2D mode\}) \\ \hline
$\mathcal{I}$ & Aggregate interference power  \\ \hline
$\mathcal{R}_{\chi}$ & Average link capacity in mode $\chi$ \\ \hline
$\mathcal{T}$ & Total network capacity \\ \hline
$\mathcal{L}_X(.)$ & Laplace transform of the {\em pdf} of  $X$   \\ \hline
$\mathbb{P}\left\{.\right\}$ & Probability of an event   \\ \hline
$\mathbb{E}\left[.\right]$ & Expectation  \\ \hline
$\mathbbm{1}_{\left\{.\right\}}$ & Indicator function   \\ \hline
$f_X(.)$ & {\em pdf} of random variable $X$   \\ \hline

{\multirow{2}{*}{$f_{X|Y}(x|y)$}} & Conditional {\em pdf} of the random variable X  \\
 & given $Y=y$   \\ 
\hline
\end{tabular}
\label{table}
\end{table}

\section{System Model, Assumptions, and Methodology of Analysis} \label{sys}

\subsection{Network Model}

We consider a D2D-enabled single-tier (i.e., macro-tier only) uplink cellular network. The BSs are spatially distributed in $\mathbb{R}^2$ according to the PPP $\mathbf{\Psi}=\left\{m_i; i=1,2,3,...\right\}$ with intensity $\lambda$, where $m_i \in \mathbb{R}^2$ is the location of the $i^{th}$ BS. The UEs (i.e., potential transmitters) are spatially distributed in $\mathbb{R}^2$ according to an independent PPP $\mathbf{\Phi}=\left\{u_i; i=1,2,3,..\right\}$\footnote{With a slight abuse of notation, we will use $m_i$ to denote both the location of the $i^{th}$ BS and the $i^{th}$ BS itself, and the same for $u_i$.} with intensity $\mathcal{U}$ such that $\mathcal{U} \gg \lambda$. We assume that each BS will always have at least one UE to serve in the uplink. All UEs have a maximum transmit power of $P_u$. In D2D mode, two UEs are able to communicate directly (i.e., in single hop). A UE can bypass the BS and communicate with its intended receiver if the receiver is located within the D2D proximity $R_{max}$. The D2D proximity $R_{max}$ is determined by the maximum transmit power $P_u$ of a UE and receiver sensitivity $\rho_{min}$. That is, $R_{max}= \left(\frac{P_u}{\rho_{min}}\right)^{\frac{1}{\eta_d}}$, where $\eta_d$ is the path-loss exponent for the D2D links\footnote{$R_{max}$ is calculated based on the radio channel model to be presented in Sec.~\ref{Rmodel}.}. A UE which has its intended receiver located within its D2D proximity is referred to as a {\em potential D2D} transmitter (or equivalently potential D2D UE). Note that a {\em potential} D2D transmitter does not necessarily select the D2D mode of communication. The communication mode is selected based on the mode selection scheme to be presented later in this paper. It is assumed that the probability of being a potential D2D transmitter is independent of the transmitter location. Therefore, the potential D2D transmitters constitute a PPP with intensity $\mathcal{D} \leq \mathcal{U}$. 

It is assumed that the intended receiver for a potential D2D transmitter is uniformly distributed inside the D2D proximity $R_{max}$. That is, the probability density function ({\em pdf}) of the D2D link distance is given by: $f_{r_d}(r)= \frac{2r}{R_{max}^2}$, $0 \leq r \leq R_{max}$. 


Due to the maximum transmit power  ($P_u$) constraint, the UEs use a truncated channel inversion power control in which the transmit power compensates the path-loss to keep the average signal power received at the intended receiver (i.e., BS or D2D receiver for cellular and D2D mode of communication, respectively) equal to certain threshold $\rho_o \geq \rho_{min}$ \cite[chapter 4]{andrea-book}. Therefore, a connection (i.e., cellular uplink or D2D link) is established if and only if the transmit power required for the path-loss inversion is less than or equal to $P_u$. Otherwise, the UE does not transmit and goes into an outage (hereafter denoted by {\em cellular truncation outage}) due to the insufficient transmit power. For the D2D UEs, the truncated channel inversion power control reduces the D2D proximity to $R=\left(\frac{P_u}{\rho_o}\right)^\frac{1}{\eta_d}$, and hence, the intensity of the potential D2D UEs reduces to $p  \mathcal{D}$, where $p=\left(\frac{R}{R_{max}}\right)^2=\left(\frac{\rho_{min}}{\rho_{o}}\right)^\frac{2}{\eta_d}$. Note that $(1-p)$ denotes the {\em D2D truncation probability}. 

Universal frequency reuse is used across the cellular network. However, there is no intra-cell interference between cellular UEs. That is, each BS assigns a unique channel to each of its associated UEs. Since all channels have similar interference statistics, we restrict our analysis to a one uplink channel which is shared by the D2D UEs in an underlay spectrum sharing fashion. 

\subsection{Radio Channel Model} \label{Rmodel}

A general power-law path-loss model is considered in which the signal power decays at the rate $r^{-\eta}$ with the propagation distance $r$, where $\eta > 2$ is the path-loss exponent. Due to the different propagation environments experienced by the cellular links and the D2D links, each type of link has its own path-loss exponent, namely, $\eta_c$ and $\eta_d$, respectively, for the cellular links and D2D links. The channel (power) gain between two generic locations $x,y \in \mathbb{R}^2$ is denoted by $h(x,y)$. All the channel gains are assumed to be independent of each other, independent of the spatial locations, symmetric, and are identically distributed (i.i.d.). Therefore, for the brevity of exposition, hereafter, the spatial  indices $x,y$ are dropped. For analysis, only Rayleigh fading environment is assumed\footnote{Techniques to relax the Rayleigh fading assumption to general fading channels can be found in \cite{our-survey}.}, hence, the channel gain $h$ is assumed to be exponentially distributed with unit mean. An ${\rm SINR}$ capture model is considered where a message can be successfully decoded at the receiver if and only if the ${\rm SINR}$ at the receiver is greater than a certain threshold $\theta$. If the ${\rm SINR}$ at the receiver does not exceed the threshold $\theta$, the link experiences an outage (hereafter denoted by {\em ${\rm SINR}$ outage}). 

\subsection{User Association, Mode Selection, and UE Classification}

To ensure reliable uplink association and avoid the ping-pong effects due to handovers, UEs associate to the BSs based on their long-term average link quality (and hence distance). That is, the UEs (i.e., potential transmitters) associate with their nearest BSs. Note that, in the D2D mode, for a transmitter UE, the receiver UE does not need to be in the same cell.  A flexible mode selection scheme based on the biased link quality is applied to tune the tradeoff between power consumption, interference, spatial frequency reuse, and  data offloading. A potential D2D transmitter chooses the D2D mode if the biased D2D link quality  is at least as good as the cellular uplink quality. That is, a potential D2D transmitter chooses the D2D mode if $T_d r_d^{-\eta_d} \geq r_c^{-\eta_c}$, where $r_d$ is the D2D link distance, $r_c$ is the distance between the UE and its closest BS (i.e., the cellular uplink distance), and $T_d$ is a bias factor to control traffic offloading from the cellular infrastructure to the D2D mode of communication. On one extreme, setting $T_d=0$ disables the D2D communication. On the other extreme, setting $T_d=\infty$ forces each potential D2D UE to communicate via D2D mode. 

The network model with the aforementioned mode selection criterion  is illustrated in Fig~\ref{network} for different values of bias factor $T_d$. It is worth mentioning that one main advantage of the aforementioned mode selection criterion is that it correlates the locations of D2D transmitters and the D2D link distances to the locations of the BSs which introduces an inherent interference protection to the cellular uplink. That is, the proposed mode selection criterion along with truncated channel inversion power control ensures that the interference from a generic D2D transmitter received at a generic BS is upper bounded by $ T_d \rho_o$, where the bias factor $T_d$ can be used to control the interference temperature at the BSs. The set of D2D transmitters will form a Poisson hole process (with random hole radii) with holes centered around the BSs \cite{martin-book, martin-cognitive}. Following the proposed mode selection scheme, the probability that a generic potential D2D UE selects the D2D mode and the intensity of D2D links are given via the following Lemma.  

\begin{lemma}
\label{lem1}
For a cellular network with BS intensity $\lambda$, D2D range $R$, and D2D link biasing $T_d$, the probability that a generic potential D2D UE selects the D2D mode is given by $ \mathcal{P}_d =\frac{ \eta_c T_d^{\frac{2}{\eta_d}}}{ \eta_d R^2} \left(\frac{1}{\pi \lambda}\right)^{\frac{\eta_c}{\eta_d}} \gamma\left(\frac{\eta_c}{\eta_d},\pi \lambda \left(\frac{R^{\eta_d}}{T_d}\right)^{\frac{2}{\eta_c}} \right)$, where $\gamma(a,b)=\int_{0}^{b} x^{a-1} e^{-x} dx$ is the lower incomplete gamma function. For equal path-loss exponents, the expression for the D2D mode selection probability reduces to $\mathcal{P}_d = \frac{T_d^{\frac{2}{\eta}}}{\pi \lambda R^2}  \left(1-\exp\left(- \pi \lambda R^2 T_d^{\frac{-2}{\eta}}\right)\right)$. The intensity of D2D links is given by $ p \mathcal{D} \mathcal{P}_d$. 
\end{lemma}

\begin{proof}
See \textbf{Appendix \ref{proof1}}.
\end{proof}

\begin{figure}[th]
	\begin{center}
		\subfigure[]{\label{nett1}\includegraphics[trim= 2cm 6cm 2cm 5.5cm, clip, width=2.6 in]{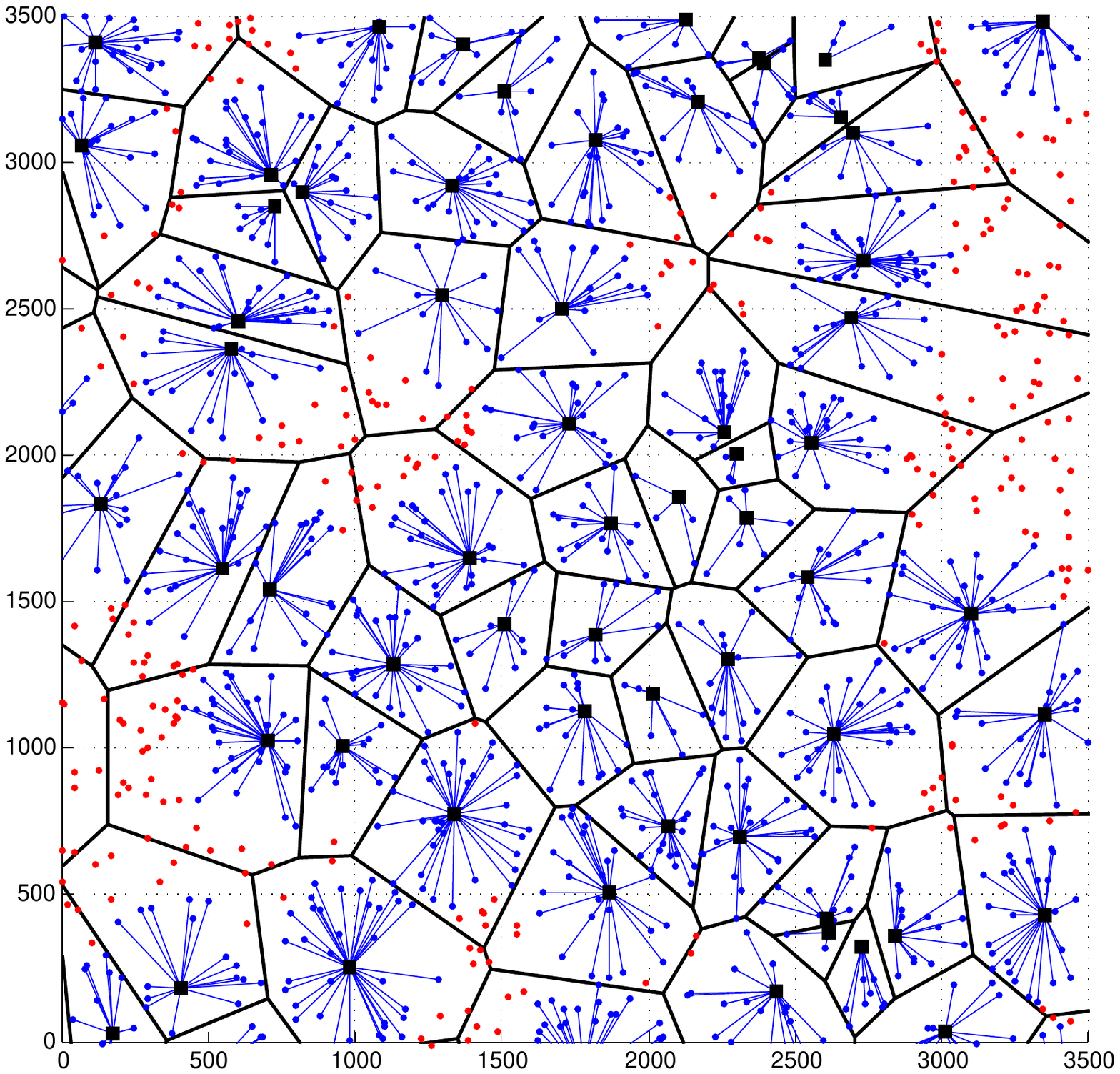}}
	  \subfigure[]{\label{nett2}\includegraphics[trim= 2cm 6cm 2cm 5.5cm, clip, width=2.6 in]{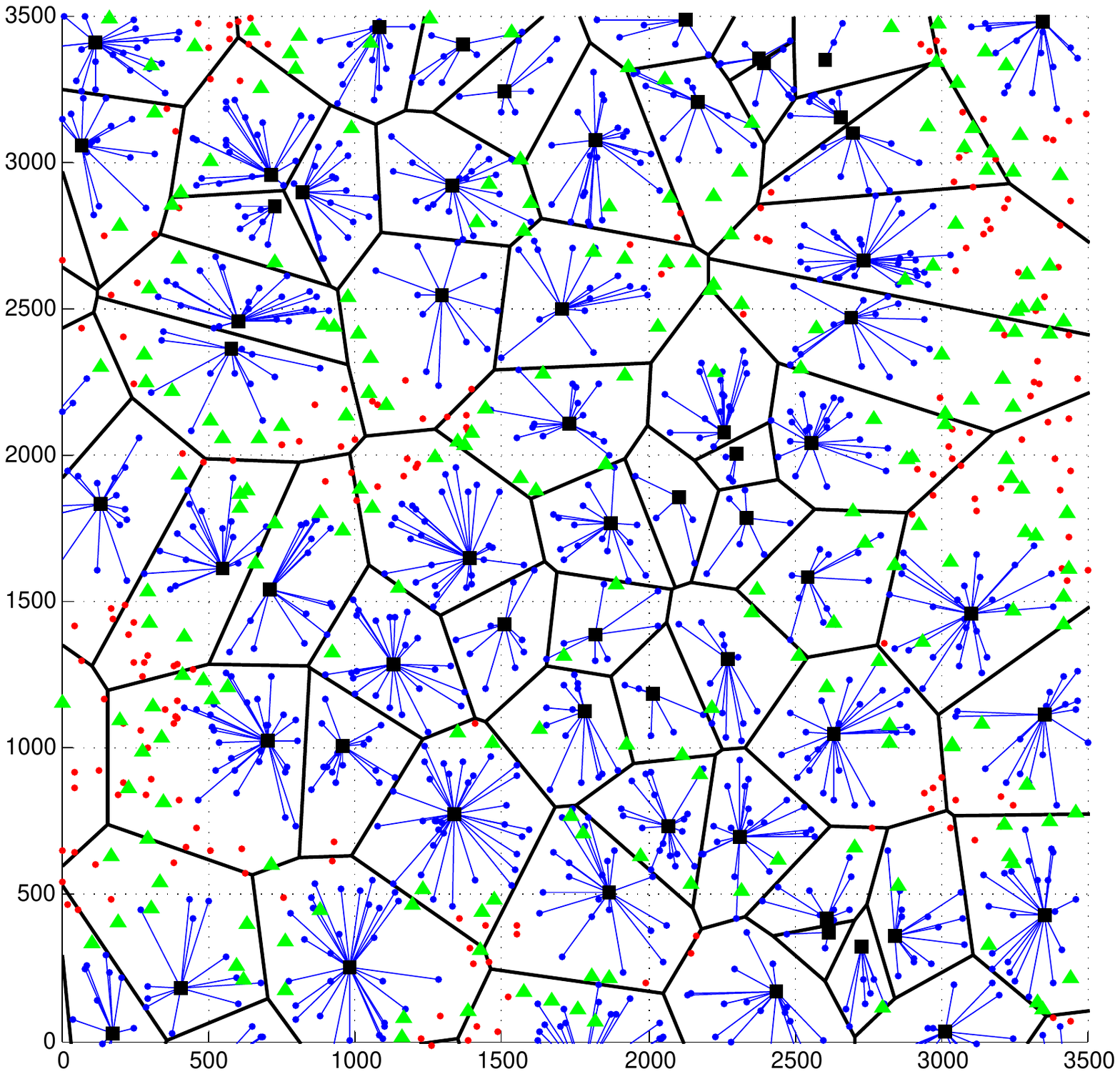}}
	\end{center}
	\caption{ The black squares represent the BSs, the blue dots represent the cellular UEs, blue lines indicate the uplink connections, red dots represent users in truncation outage, and the green triangles represent the D2D transmitters (D2D receivers are omitted for clarity of presentation). The network model is simulated for $\lambda = 5$ BS/km$^2$, $\mathcal{D} = \mathcal{U}=100$ UE/km$^2$, $\rho_{min}$ = -90, $\rho_o = -80$ dBm, and (a) $T_d=0$, (b) $T_d =1$.}
	\label{network}
\end{figure}


Due to truncated channel inversion power control, not all of the UEs can communicate in the uplink. That is, UEs located at a distance greater than $\left(\frac{P_u}{\rho_o}\right)^\frac{1}{\eta_c}$ from their nearest BS are not able to communicate in the cellular mode due the limited transmit power. Therefore, the UEs are divided into two subsets, namely, the covered\footnote{We use the term {\em covered} to denote that a UE does not experience truncation outage and can be served in the uplink by at least one BS.} UEs and the uncovered UEs. Furthermore, D2D communication divides the set of UEs into two other independent subsets, namely, the subset of potential D2D UEs and the set of non-potential D2D UEs. Therefore, the complete set of UEs $\mathbf{\Phi}$ is divided into four non-overlapping subsets as shown in Fig.~\ref{users}. In the first case, UEs are neither covered by the BSs nor are potential D2D UEs, hence UEs in case $\#1$ will not transmit and will experience truncation outage due to insufficient transmit power. UEs in case $\#2$ are covered by the BSs but are not potential D2D UEs. On the contrary, UEs in case $\#3$ are uncovered by the BSs but are potential D2D UEs. Therefore,  UEs in case $\#2$ and case $\# 3$ are forced to communicate via, cellular mode and D2D mode, respectively. Only UEs in case $\#4$ are covered by the BSs and are potential D2D transmitters, and hence, only UEs in case $\# 4$ have the opportunity to apply the aforementioned selection criterion to select their operation mode (i.e., cellular mode or D2D mode). Fig.~\ref{users} shows the classification of UEs, connection type, the type of point process they constitute, and their intensity.

\begin{figure}[th]
	\begin{center}
\includegraphics[width=2.6 in]{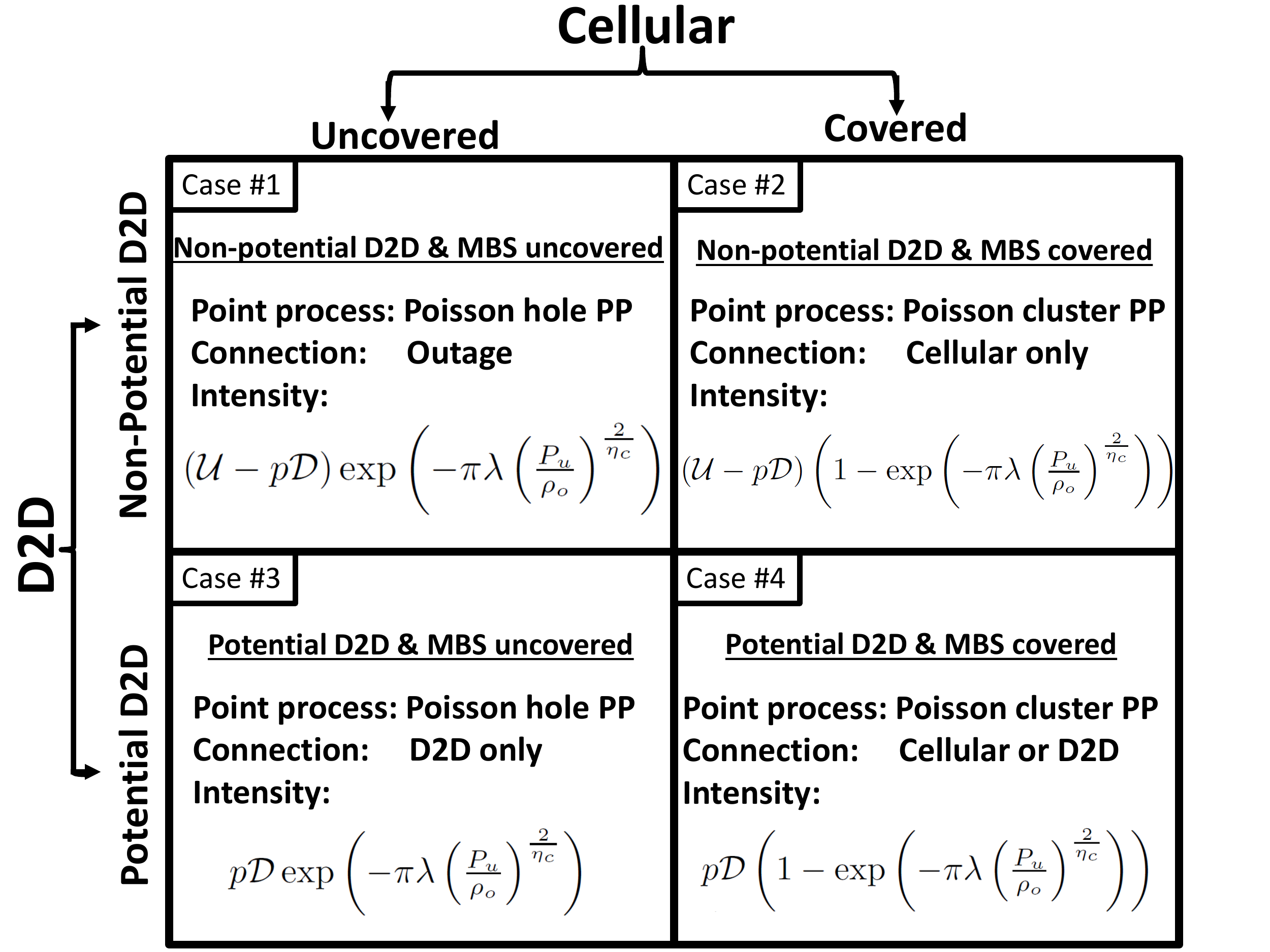}
	\end{center}
	\caption{Classification of UEs (MBS stands for `macro BS').}
	\label{users}
\end{figure}

\subsection{Methodology of Analysis and Performance Metrics}

Due to the assumed power control along with the random locations of the BSs and UEs, the transmit powers of the UEs and the ${\rm SINR}$s experienced by the receivers are random. First, we characterize the transmit powers of the users in each of the aforementioned cases. We characterize the transmit power via its probability density function ({\em pdf}) and its $\alpha$th moment (for $\alpha > 0$). Then, we characterize the ${\rm SINR}$ by deriving its cumulative distribution function ({\em cdf}). Having the transmit powers and the ${\rm SINR}$s characterized, several insights into the network performance can be obtained. In this paper, the main performance metrics are the {\em ${\rm SINR}$ outage probability}, the {\em link capacity}, and the {\em total network capacity}. The link capacity per unit bandwidth is obtained by using Shannon's formula. The total network capacity is the sum capacity of all operating links normalized per unit area and it reflects the spatial frequency reuse efficiency. 


\section{Analysis of Transmit Power of UEs} \label{power}

Due to the random network topology along with the truncated channel inversion power control used by the UEs, each UE will have different transmit power which depends on the operation mode (i.e., cellular or D2D) as well as the link distance. In this section, we derive the {\em pdf} as well as the moments of the transmit powers of the UEs for each of the cases shown in Fig.~\ref{users}. Note that UEs in case $\# 1$ are not transmitting. Hence, the transmit power of a generic user in case $\#1$ is $P_1 = 0$.

\subsection{D2D Mode}


A UE selects the D2D mode if $T_d r_d^{-\eta_d} \geq r_c^{-\eta_c}$. Note that only UEs in case $\# 3$ and case $\# 4$ can communicate in the D2D mode.  Let $\tilde{r}_d \eqdef \left\{r_d : T_d r_d^{-\eta_d} \geq  r_c^{-\eta_c}\right\}$ denote the conditional D2D link distance of a UE operating in the D2D mode (i.e., conditioning on the mode selection). Then, due to the applied selection criterion, we have $\tilde{r}_d^{\eta_d} \leq  T_d r_c^{\eta_c}$. Using this fact, the transmit power of a generic UE operating in the D2D mode can be written as $P_d = \rho_o \tilde{r}_d^{\eta_d} = \left\{\rho_o r_d^{\eta_d} : r_d^{\eta_d} \leq T_d r_c^{\eta_c} \right\}$, and can be characterized via the following lemma.

\begin{lemma}
\label{lem2}
In a single-tier Poisson D2D-enabled cellular network with truncated channel inversion power control with cutoff threshold $\rho_o$ and bias factor $T_d$, the {\em pdf} of the transmit power of a generic UE operating in the D2D mode is given by

\small
\begin{align}
f_{P_d}(x)&= \frac{2 x^{\frac{2}{\eta_d}-1}  (\pi \lambda)^{\frac{\eta_c}{\eta_d}} e^{-\pi \lambda \left(\frac{x}{T_d \rho_o}\right)^\frac{2}{\eta_c}} }{\eta_c (T_d \rho_o)^{\frac{2}{\eta_d}} \gamma\left(\frac{\eta_c}{\eta_d},\pi \lambda \left(\frac{P_u}{T_d\rho_o}\right)^{\frac{2}{\eta_c}}\right)}, \notag  \\
& \quad \quad \quad \quad  \quad \quad  \quad \quad \quad \quad   \quad \quad   \quad \quad 0 \leq x \leq P_u.
\end{align}
\normalsize

The moments of the transmit power can be obtained as 
\small
\begin{align}
\mathbb{E}\left[{P_d^\alpha}\right] &=\frac{ (T_d \rho_o)^\alpha   \gamma\left(\frac{\alpha \eta_c}{2}+ \frac{\eta_c}{\eta_d},\pi \lambda \left(\frac{P_u}{T_d\rho_o}\right)^{\frac{2}{\eta_c}}\right) }{(\pi \lambda)^\frac{\alpha\eta_c}{2} \gamma\left(\frac{\eta_c}{\eta_d},\pi \lambda \left(\frac{P_u}{T_d\rho_o}\right)^{\frac{2}{\eta_c}}\right)}.
\end{align}
\normalsize
\end{lemma}

\begin{proof}\hspace{-1mm}:
See \textbf{Appendix \ref{proof2}}.
\end{proof}

\subsection{Cellular Mode}

There are  two cases where a UE operates in the cellular mode. The first case is that the UE is not a potential D2D and $P_u$ is sufficient to invert the path-loss towards the nearest BS such that the received power at the BS is equal to $\rho_o$ (i.e., UEs in case $\# 2$). The second case is that the UE is a potential D2D UE and its uplink quality towards the serving BS is better than the biased link quality towards the D2D receiver (i.e., case $\#4$ when $r_c^{-\eta_c} \geq T_d r_d^{-\eta_d}$). Note that UEs in case $\#2$ have no option except to communicate in the cellular mode. Let $\tilde{r}_c \eqdef \left\{r_c : r_c \leq \left(\frac{P_u}{r_o}\right)^\frac{1}{\eta_c}\right\}$ denote the conditional cellular distance of a UE in case $\# 2$ (i.e., conditioning on that the UE is not in truncation outage). Then, the transmit power of a generic UE in case $\# 2$, $P_2 = \rho_o \tilde{r}_c^{\eta_c} = \left\{\rho_o r_c^{\eta_c} : r_c^{\eta_c} \leq \frac{P_u}{r_o}\right\}$, can be characterized via the following lemma.

\begin{lemma}
\label{lem3}
In a single-tier Poisson D2D-enabled cellular network with truncated channel inversion power control with cutoff threshold $\rho_o$, the {\em pdf} of the uplink transmit power of a generic covered non-potential UE (i.e., in case $\#2$) is given by
 
\begin{align}
f_{P_2}(x) &=  \frac{2 \pi \lambda x^{\frac{2}{\eta_c}-1} e^{-\pi \lambda \left(\frac{x}{\rho_o}\right)^{\frac{2}{\eta_c}}}}{\eta_c  \rho_o^{\frac{2}{\eta_c}} \left(1-e^{-\pi \lambda \left(\frac{P_u}{\rho_o}\right)^{\frac{2}{\eta_c}}}\right)}, 0\leq x \leq P_u. 
\label{pp2}
\end{align}
The moments of the transmit power can be obtained as 
\begin{equation}
\mathbb{E}\left[P_2^\alpha\right] =   \frac{\rho_o^\alpha \gamma\left(\frac{\alpha \eta_c}{2}+1, \pi \lambda \left(\frac{P_u}{\rho_o}\right)^\frac{2}{\eta_c}\right)}{(\pi \lambda)^{\frac{\alpha \eta_c}{2}}\left(1-e^{-\pi \lambda \left(\frac{P_u}{\rho_o}\right)^{\frac{2}{\eta_c}}}\right)}.
\end{equation}  
\end{lemma}

\begin{proof}\hspace{-1mm}:
According to the PPP assumption of the locations of the BSs, the cellular link distance $r_c$ is Rayleigh distributed with the {\em pdf} $f_{r_c}(r)=2 \pi \lambda r e^{-\lambda \pi r^2}$  $0 \leq r$ \cite{trac}. Due to the truncated channel inversion power control, the transmit power of the cellular UE should be $P_2 = \rho_o r_c^{\eta_c}$ and $P_2 \leq P_u$. Hence, the {\em pdf} of the transmit power can be obtained as $f_{P_2}(x) = \frac{\frac{2 \pi \lambda x^{\frac{2}{\eta_c}-1} e^{-\pi \lambda \left(\frac{x}{\rho_o}\right)^{\frac{2}{\eta_c}}}}{\eta_c  \rho_o^{\frac{2}{\eta_c}}}}{\int_0^{P_u}\frac{2 \pi \lambda y^{\frac{2}{\eta_c}-1} e^{-\pi \lambda \left(\frac{y}{\rho_o}\right)^{\frac{2}{\eta_c}}}}{\eta_c  \rho_o^{\frac{2}{\eta_c}}}dy}$, $0 \leq x \leq \infty$, which gives (\ref{pp2}). The moments of $P_2$ is obtained as $\int_0^{P_u} x^\alpha f_{P_2}(x) dx$.

\end{proof}

Let $\tilde{P}_4$ be the conditional transmit power for a potential D2D UE operating in the cellular mode (i.e., conditioning on the mode selection for a UE in case $\#4$), and $\tilde{\tilde{r}}_c \eqdef \left\{\tilde{r}_c : \tilde{r}_c^{-\eta_c} \leq T_d r_d^{-\eta_d}\right\}$ be the conditional cellular link distance of a UE in case $\# 4$ operating in the cellular mode\footnote{The notation $\tilde{\tilde{r}}_c$ is selected to reflect the two conditions for the considered cellular link distance, namely, conditioning on being in case $\#4$ and conditioning on the cellular mode operation.}. The transmit power $\tilde{P}_4 = \rho_o \tilde{\tilde{r}}_c^{\eta_c} = \left\{\rho_o r_c^{\eta_c} : r_c^{\eta_c} \leq \frac{P_u}{\rho_o} \cap r_c^{-\eta_c} > T_d r_d^{-\eta_d} \right\} $ can be characterized via the following lemma.

\begin{lemma}
\label{lem4}
In a single-tier Poisson D2D-enabled cellular network with truncated channel inversion power control with cutoff threshold $\rho_o$ and bias factor $T_d$, the {\em pdf} of the transmit power of a generic potential D2D UE (i.e., in case $\#4$), which is operating in the cellular mode, is given by
 
\small
\begin{align}
f_{\tilde{P}_4}(x) & =  \frac{2 \pi \lambda x^{\frac{2}{\eta_c}-1} \left({P}_u^\frac{2}{\eta_d} - (T_dx)^\frac{2}{\eta_d}\right) e^{-\pi \lambda \left(\frac{x}{\rho_o}\right)^{\frac{2}{\eta_c}}}}{\mathbb{P}\left\{\tilde{r}_c^{\eta_c} < \frac{r_d^{\eta_d}}{T_d}\right\} \eta_c  \rho_o^{\frac{2}{\eta_c}} P_u^\frac{2}{\eta_d} \left(1-e^{-\pi \lambda \left(\frac{P_u}{\rho_o}\right)^{\frac{2}{\eta_c}}}\right)}, \notag \\
& \quad \quad \quad \quad  \quad \quad  \quad \quad \quad \quad   \quad \quad   \quad \quad 0 \leq x \leq \frac{P_u}{\max(T_d,1)} 
\end{align}
\normalsize
where 

\small
\begin{align}
\mathbb{P}\left\{\tilde{r}_c^{\eta_c} < \frac{r_d^{\eta_d}}{T_d}\right\}&= 1 - \min\left(1,T_d^\frac{2}{\eta_d}\right) + \frac{\min\left(1,T_d^\frac{2}{\eta_d}\right)}{{\left(1-e^{-\pi \lambda \left(\frac{P_u}{\rho_o}\right)^{\frac{2}{\eta_c}}}\right)}}\notag \\
&- \frac{\eta_c(T_d \rho_o)^{\frac{2}{\eta_d}} \gamma\left(\frac{\eta_c}{\eta_d},\pi \lambda \left(\frac{P_u}{\max(T_d,1)\rho_o}\right)^{\frac{2}{\eta_c}}\right)}{\eta_d P_u^\frac{2}{\eta_d } (\pi  \lambda)^{\frac{\eta_c}{\eta_d}} {\left(1-e^{-\pi \lambda \left(\frac{P_u}{\rho_o}\right)^{\frac{2}{\eta_c}}}\right)} }.  
\label{pp3}
\end{align}

\normalsize
The moments of the transmit power $\tilde{P}_4$ can be obtained as 

\small
\begin{align}
\mathbb{E}[{\tilde{P}_4}^\alpha] &= \frac{1}{\mathbb{P}\left\{\tilde{r}_c < \frac{r_d}{T_d}\right\}}\left[\frac{\rho_o^\alpha \gamma\left(\frac{\alpha \eta_c}{2}+1, \pi \lambda \left(\frac{P_u}{\max(T_d,1) \rho_o}\right)^\frac{2}{\eta_c}\right)}{(\pi \lambda)^{\frac{\alpha \eta_c}{2}}\left(1-e^{-\pi \lambda \left(\frac{P_u}{\rho_o}\right)^{\frac{2}{\eta_c}}}\right)} \right. \notag \\ 
&\left.-\frac{T_d^\frac{2}{\eta_d} \rho_o^{\alpha+\frac{\eta_d}{2}} \gamma\left(\frac{\eta_c}{\eta_d}+ \frac{\alpha \eta_c}{2}+1, \pi \lambda \left(\frac{P_u}{\max(T_d,1) \rho_o}\right)^\frac{2}{\eta_c}\right)}{(\pi \lambda)^{\frac{\eta_c}{\eta_d}+\frac{\alpha \eta_c}{2}} P_u^\frac{2}{\eta_d} \left(1-e^{-\pi \lambda \left(\frac{P_u}{\rho_o}\right)^{\frac{2}{\eta_c}}}\right)}\right].
\end{align}
\normalsize


\end{lemma}

\begin{proof}\hspace{-1mm}:
See \textbf{Appendix \ref{proof4}}.
\end{proof}
We end this section with the following corollary, which gives the moments of the transmit power of a generic UE operating in the cellular mode.

\begin{corollary}
\label{col1}
The moments of the transmit power of a generic user operating in the cellular mode is given by
\small
\begin{align}
\mathbb{E}[P_c^\alpha] 
&= (1- \frac{p \mathcal{P}_d \mathcal{D}}{\mathcal{U}} ) \mathbb{E}\left[ P_2^\alpha\right]  +  \frac{p \mathcal{P}_d \mathcal{D}}{\mathcal{U}} \mathbb{E}\left[ \tilde{P}_4^\alpha\right]
\end{align}
\normalsize
where $\mathbb{E}\left[ P_2^\alpha\right]$ and $\mathbb{E}\left[ \tilde{P}_4^\alpha\right]$ are given in \textbf{Lemma~\ref{lem3}} and \textbf{Lemma~\ref{lem4}}, respectively. 
\end{corollary}
\begin{proof}\hspace{-1mm}:
Since the two events (i.e., being a non-potential cellular UE [i.e., in case $\#2$] and being a potential cellular UE [i.e., in case $\#4$]) are mutually exclusive, the transmit power of a generic user in the cellular mode can be expressed as $P_c =\mathbbm{1}_{\left\{\text{user is in case $\#2$}\right\}} P_2 + \mathbbm{1}_{\left\{\text{user is in case $\#4$}\right\}} \tilde{P}_4$. Hence, $\mathbb{E}[P_c^\alpha] = \mathbb{E}\left[\left(\mathbbm{1}_{\left\{\text{user is in case $\#2$}\right\}} P_2 + \mathbbm{1}_{\left\{\text{user is in case $\#4$}\right\}} \tilde{P}_4\right)^\alpha\right]$. Exploiting the mutually exclusive property of the two events, we have $\mathbb{E}[P_c^\alpha]= \mathbb{P}\left\{\text{user is in case $\#2$}\right\} \mathbb{E}\left[ P_2^\alpha\right] + \mathbb{P}\left\{\text{user is in case $\#4$}\right\} \mathbb{E}\left[\tilde{ P}_4^\alpha\right] $ which proves the corollary. 
\end{proof}

\section{Analysis of ${\rm SINR}$} \label{ana}

Consider an arbitrary UE operating in mode $\chi \in \left\{c,d\right\}$, where $c$ and $d$ denote the cellular mode and D2D mode, respectively. The {\rm SINR} experienced at the receiver associated with the test UE, which can be located in an arbitrary location $y \in \mathbb{R}^2$, can be written as
\footnotesize
\beq
{\rm SINR}_\chi = \frac{ \rho_o h_o}{\sigma^2+ \underbrace{\sum_{u_i\in \tilde{\mathbf{\Phi}}_c} P_{c_i} h_i \left\|y - u_i\right\|^{-\eta_c} }_{\mathcal{I}_{c\chi}}+ \underbrace{\sum_{u_j\in \tilde{\mathbf{\Phi}}_d} P_{d_j} h_j \left\|y-u_j\right\|^{-\eta_d}}_{\mathcal{I}_{d\chi}} }
\eeq
\normalsize
where the intended signal is always equal to $\rho_o h_o$ due to the employed power control. The random variables $\mathcal{I}_{c\chi}$ and $\mathcal{I}_{d\chi}$ denote the aggregate interference caused to the link under investigation (which operates in mode $\chi$) from the UEs operating in the cellular mode (who are represented by the PP $\tilde{\mathbf{\Phi}}_c$) and the UEs operating in the D2D mode (who are represented by the PP $\tilde{\mathbf{\Phi}}_d$), respectively. 
The ${\rm SINR}$ outage probability can be calculated as

\begin{align}
\mathbb{P}\left\{{\rm SINR}_\chi \leq \theta \right\} &= \mathbb{P} \left\{\rho_o h_o \leq \theta \left(\sigma^2+ \mathcal{I}_{c\chi} + \mathcal{I}_{d\chi} \right) \right\} \notag \\
 &= 1-\exp\left\{ -\frac{\theta}{\rho_o} \left(\sigma^2+ \mathcal{I}_{c\chi} + \mathcal{I}_{d\chi} \right) \right\} \notag \\
&= 1-\exp\left\{ -\frac{\theta}{\rho_o} \sigma^2\right\} \mathcal{L}_{\mathcal{I}_{c\chi}}\left(\frac{\theta}{\rho_o}\right) \mathcal{L}_{\mathcal{I}_{d\chi}}\left(\frac{\theta}{\rho_o}\right) 
\label{SINR1}
\end{align}
where the second equality follows from the exponential distribution of $h_o$ and $\mathcal{L}_X(.)$ denotes the Laplace transform\footnote{Hereafter, $\mathcal{L}_X(.)$ will be denoted as the LT of $X$.} (LT) of the {\em pdf} of the random variable $X$. 

To find the exact distribution of the $\rm SINR$ we have to find the LT for the interference experienced by the test receiver from both the PPs $\tilde{\mathbf{\Phi}}_c$ and $\tilde{\mathbf{\Phi}}_d$. Although we have assumed that the complete set of UEs is modeled via a PPP, as has been discussed earlier, neither $\tilde{\mathbf{\Phi}}_c$ nor $\tilde{\mathbf{\Phi}}_d$ is a PPP. In particular, the set of interfering cellular UEs $\tilde{\mathbf{\Phi}}_c$ can be modeled via a softcore process (e.g., Strauss process) and the set of interfering D2D UEs $\tilde{\mathbf{\Phi}}_d$ can be modeled via a Poisson hole process with random radii of the holes. Note that although the set of all cellular UEs on all channels constitute a Poisson cluster process (i.e., clustered around the BSs due to the truncation outage probability), the set of interfering UEs (i.e., UEs operating on the same channel) constitutes a softcore point process due to the orthogonal frequency assignment within the same cell. Hence, only one UE is allowed to use a given frequency channel within the Voronoi cell of its serving BS. Unfortunately, neither the Poisson hole process nor the 
softcore point process is analytically tractable \cite{martin-book, pp-cellular}. For analytical tractability, we approximate each of the point processes with an equi-dense PPP. Note that in our approximation we will ignore the mutual correlations between the positions of the simultaneously active UEs. However, the correlations between the interfering UEs and the test transmitter and receiver are captured by accounting for the interference protection induced  by the system model. The assumptions are stated formally as follows.     

\begin{assumption}
The set of interfering UEs operating in the cellular mode ($\tilde{\mathbf{\Phi}}_c$) constitutes a PPP and that the transmit powers of the UEs are independent.
\label{ass1}
\end{assumption}
Note that, as discussed in \cite{and-uplink}, the dependence between the neighboring Voronoi cells imposes a weak correlation among the cellular link distances, and hence, transmit powers of the UEs are correlated due to the power control policy.  

\begin{assumption}
The set of interfering UEs operating in the D2D mode ($\tilde{\mathbf{\Phi}}_d$) constitutes a PPP.
\label{ass2}
\end{assumption}

Exploiting \textbf{Assumptions~\ref{ass1}} and \textbf{\ref{ass2}}, simple yet accurate approximations for the distribution of ${\rm SINR}$ can be obtained. It is worth mentioning that the PPP approximation for different point processes was shown to be accurate in estimating the aggregate interference and outage, if the density of the PPP and the interference exclusion region around the test receiver are carefully estimated \cite{our-survey, letter, our-MHCPP, pp-cellular}. The accuracy of \textbf{Assumptions~\ref{ass1}} and \textbf{\ref{ass2}} will be validated in Sec.~\ref{res}.

From the aforementioned assumptions, the outage probability for a generic D2D link is given by the following theorem.

\begin{theorem} \label{th1}
In a single-tier D2D-enabled cellular network with BS intensity $\lambda$, D2D biasing factor $T_d$, and truncated channel inversion power control with cutoff threshold $\rho_o$, the outage probability for a generic D2D UE is given by equation~(\ref{d2dout}).
\end{theorem}
\begin{proof}\hspace{-1mm}:
See \textbf{Appendix~\ref{pth1}}.
\end{proof}
The outage probability for a generic cellular uplink is given by the following theorem.

\begin{theorem} \label{th2}
In a single-tier D2D-enabled cellular network with BS intensity of $\lambda$, D2D biasing factor $T_d$, and truncated channel inversion power control with cutoff threshold $\rho_o$, the outage probability for a generic cellular UE is given by equation~(\ref{cellout1}).
\end{theorem}
\begin{proof}\hspace{-1mm}:
See \textbf{Appendix~\ref{pth2}}.
\end{proof}

Note that equation~(\ref{cellout1}) can be reduced to a closed-form for integer values of $\eta_c$. For instance, with $\eta_c=4$, equation~(\ref{cellout1}) reduces to equation~(\ref{cellout2}).

\begin{figure*}
\small
\begin{align}
\mathbb{P}\left\{{\rm SINR_d} \leq \theta \right\}  = 1- \exp\left( -\frac{\theta}{\rho_o} \sigma^2 - \frac{\pi \tilde{\mathcal{U}}_d}{\left|\mathbf{S}\right|} \left( \frac{\theta}{\rho_o}\right)^{\frac{2}{\eta_d}} \mathbb{E}\left[P_d^\frac{2}{\eta_d}\right] \Gamma\left(1+\frac{2}{\eta_d}\right) \Gamma\left(1-\frac{2}{\eta_d}\right) -  { \pi  \lambda \left(\frac{\theta}{\rho_o}\right)^\frac{2}{\eta_d}} \mathbb{E}\left[\tilde{P}_c^{\frac{2}{\eta_d}}\right]  \Gamma\left(1+\frac{2}{\eta_d}\right)  \Gamma\left(1-\frac{2}{\eta_d}\right)\right). 
\label{d2dout}
\end{align}

\begin{align}
\mathbb{P}\left\{{\rm SINR_c} \leq \theta \right\}  =  1- \exp\left(-\frac{\theta}{\rho_o} \sigma^2 - \frac{2 \pi \tilde{\mathcal{U}}_d}{\left|\textbf{S}\right|} \left(\frac{\theta}{\rho_o}\right)^\frac{2}{\eta_c} \mathbb{E}\left[ P_{d}^\frac{2}{\eta_c}\right] \int_{\left(\frac{1}{{\theta}T_d}\right)^\frac{1}{\eta_c}}^{\infty} \left( \frac{y }{y^{\eta_c}+1  }\right) dy - {2 \pi   \lambda} {\left(\frac{\theta}{\rho_o}\right)^\frac{2}{\eta_c}} \mathbb{E}\left[ P_c^\frac{2}{\eta_c}\right] \int_{(\theta )^\frac{-1}{\eta_c}}^{\infty} \frac{x}{x^{\eta_c}+1 }dx  \right). 
\label{cellout1}
\end{align}
\normalsize

\small
\begin{align}
\mathbb{P}\left\{{\rm SINR_c} \leq \theta \right\}  \stackrel{\eta_c=4}{=}  1- \exp\left(-\frac{\theta}{\rho_o} \sigma^2 - \frac{ \pi \tilde{\mathcal{U}}_d}{\left|\textbf{S}\right|} \sqrt{\frac{\theta}{\rho_o}} \mathbb{E}\left[ \sqrt{P_{d}}\right] \arctan\left(\sqrt{\theta T_d }\right)  - { \pi   \lambda}{\sqrt{\frac{\theta}{\rho_o}}} \mathbb{E}\left[ \sqrt{P_c}\right]  \arctan(\sqrt{\theta })\right).
\label{cellout2}
\end{align}

\normalsize

\end{figure*}

The average link capacity for a generic UE operating in mode $\chi$ can be obtained as follows:
\begin{align}
\mathcal{R}_\chi &= \mathbb{E}[\ln\left(1+{\rm SINR}_\chi \right)] \notag \\
&\stackrel{(i)}{=} \int_0^{\infty} \mathbb{P}\left\{\ln\left(1+{\rm SINR}_\chi \right) > t\right\} dt \notag \\
&= \int_0^{\infty} \mathbb{P}\left\{{\rm SINR}_\chi > \left(e^t-1\right)\right\} dt \notag \\
&\stackrel{(ii)}{=} \int_0^{\infty} \exp\left\{ -\frac{\left(e^t-1\right)}{\rho_o} \sigma^2\right\} \mathcal{L}_{\mathcal{I}_{c\chi}}\left(\frac{\left(e^t-1\right)}{\rho_o}\right) \notag \\ 
& \quad \quad \quad \quad \mathcal{L}_{\mathcal{I}_{d\chi}}\left(\frac{\left(e^t-1\right)}{\rho_o}\right) dt 
\label{rate}
\end{align}
where $(i)$ follows because $\ln\left(1+{\rm SINR}\right)$ is a strictly positive random variable. Note that the LTs of the aggregate interference experienced by a generic D2D link and a generic cellular link are derived in \textbf{Appendix~\ref{pth1}} and \textbf{Appendix~\ref{pth2}}, respectively. 

\section{Results and Discussions} \label{res}

\subsection{System Parameters and Model Validation}

We first validate our model by simulations and then present some numerical results for a D2D-enabled cellular network. Unless otherwise stated, we set the BS intensity to $\lambda=5$ BSs/km$^2$, the UE intensity to $\mathcal{U}$= 50 UE/km$^2$ where $50 \%$ of the UEs are potential D2D transmitters, the maximum transmit power $P_u=1$ W, the receiver sensitivity $\rho_{min} = -90$ dBm, the cutoff threshold $\rho=-70$ dBm, the path-loss exponents $\eta_c=\eta_d=4$, the D2D bias factor $T_d =1$, the ${\rm SINR}$ threshold $\theta=1$, $\sigma^2 = -90$ dBm, and the number of channels $\left|\mathbf{S}\right|=1$.

 In the simulation setup, we realize a PPP cellular network with intensity $\lambda$ in a 400 km$^2$ area. Then, we realize the UEs in the simulation area. In each simulation scenario, we first schedule the realized UEs according to the following criterion. A non-potential D2D is scheduled for cellular uplink transmission if and only if, (a) it can invert its channel to the serving (i.e., nearest) BS, and (b) there is no other scheduled cellular user in the same Voronoi cell. A potential D2D transmitter that selects the D2D mode of operation is always scheduled for D2D transmission. However, if a potential D2D transmitter selects the cellular mode, it is scheduled for transmission according to the  criterion for selection of cellulatr mode (i.e., if the aferomentioned conditions (a) and (b) are satisfied). After scheduling all of the realized UEs, if there are still idle BSs (i.e., BSs which are not serving any cellular users), to ensure that the saturation condition is satisfied, we introduce UEs randomly and uniformly over the simulation area until all idle BSs are active (i.e., each BS has a scheduled UE for which $P_u$ is sufficient to invert its path-loss so that it can communicate in the uplink). All UEs employ the channel inversion power control. The simulation is repeated 10000 times. 

 Fig.~\ref{valid} validates our analysis and shows that the derived model accurately captures the ${\rm SINR}$ outage. This figure manifests that a PPP with the proper intensity and exclusion areas can approximate the interference generated by different PPs with correlations among the transmitters. The figure also shows the effect of power control cutoff threshold $\rho_o$ on the ${\rm SINR}$ outage probability. The ${\rm SINR}$ outage decreases with increasing $\rho_o$ for two reasons. First, the intensity of scheduled D2D UEs decreases with increasing $\rho_o$ due to the D2D truncation probability. Hence, the interference coming from the D2D UEs  decreases, which decreases the ${\rm SINR}$ outage for all scheduled UEs (i.e., D2D and cellular UEs). Second, a higher $\rho_o$ implies a stronger desired signal power with respect to the interference power, and hence, a ${\rm SINR}$ outage probability. Note that increasing the cutoff threshold also increases average transmitted power by the UEs due to the channel inversion power control. However, the contribution of increased $\rho_o$ towards the useful signal power dominates its contribution towards the aggregate interference power, and therefore, it improves the overall ${\rm SINR}$ outage probability.

\begin{figure}[th]
	\begin{center}
 \includegraphics[width=2.6 in]{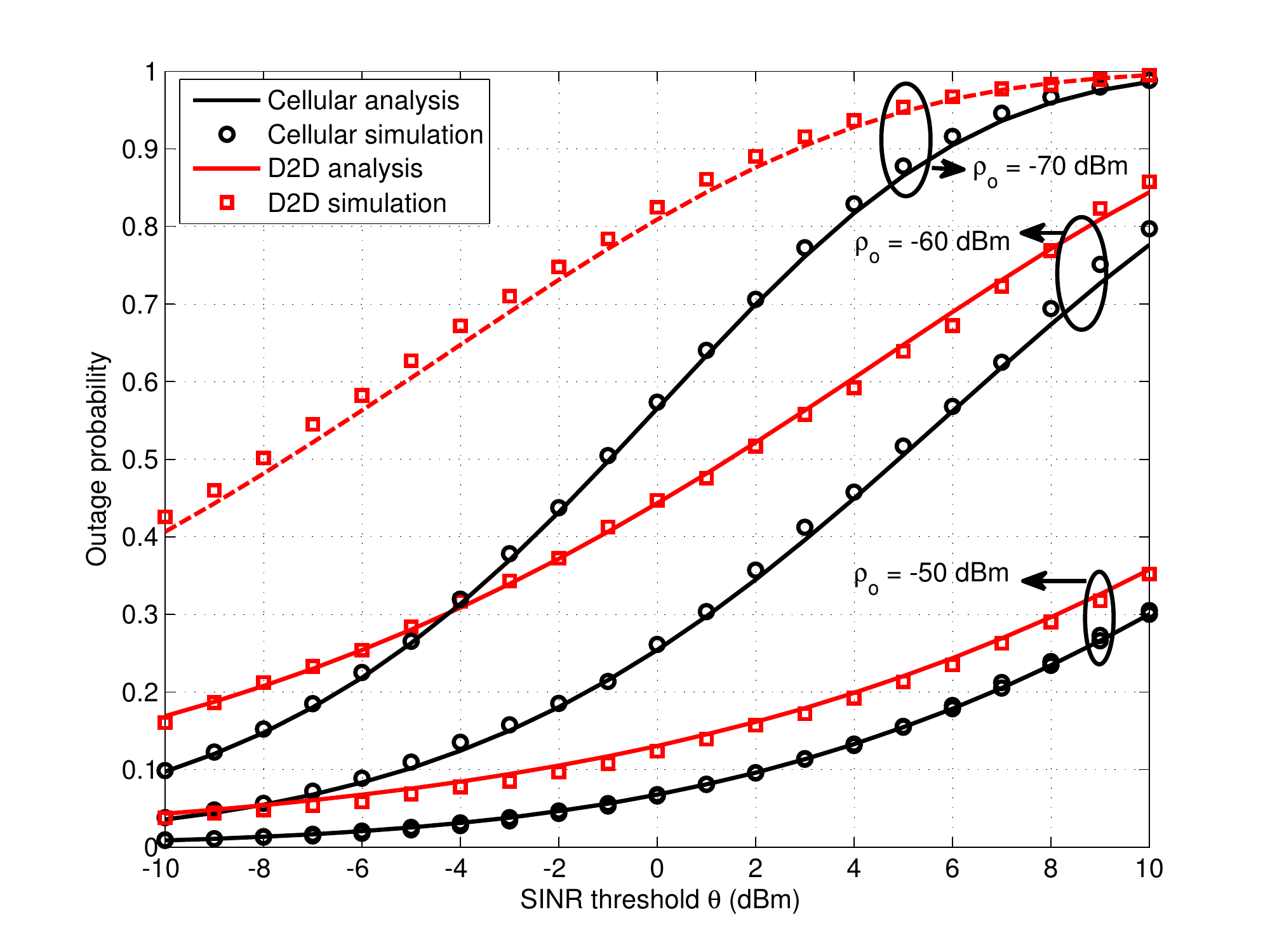}
	\end{center}
	\caption{ Model validation.}
	\label{valid}
\end{figure}

\subsection{Proposed Scheme vs. Distance-Based Mode Selection Scheme}

\begin{figure}[th]
	\begin{center}
		\subfigure[]{\label{co1}\includegraphics[width=2.6 in]{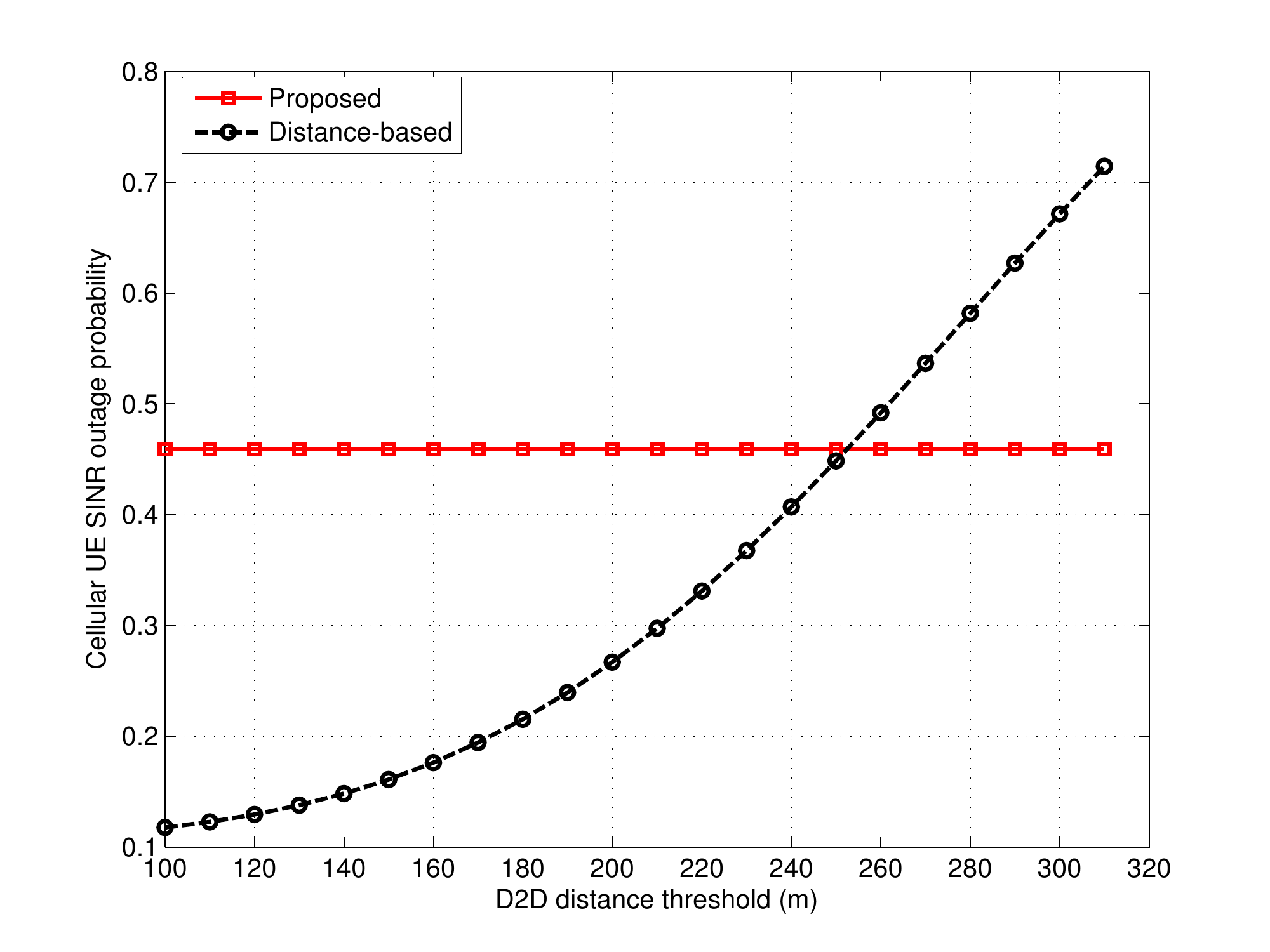}}
	  \subfigure[]{\label{co2}\includegraphics[width=2.6 in]{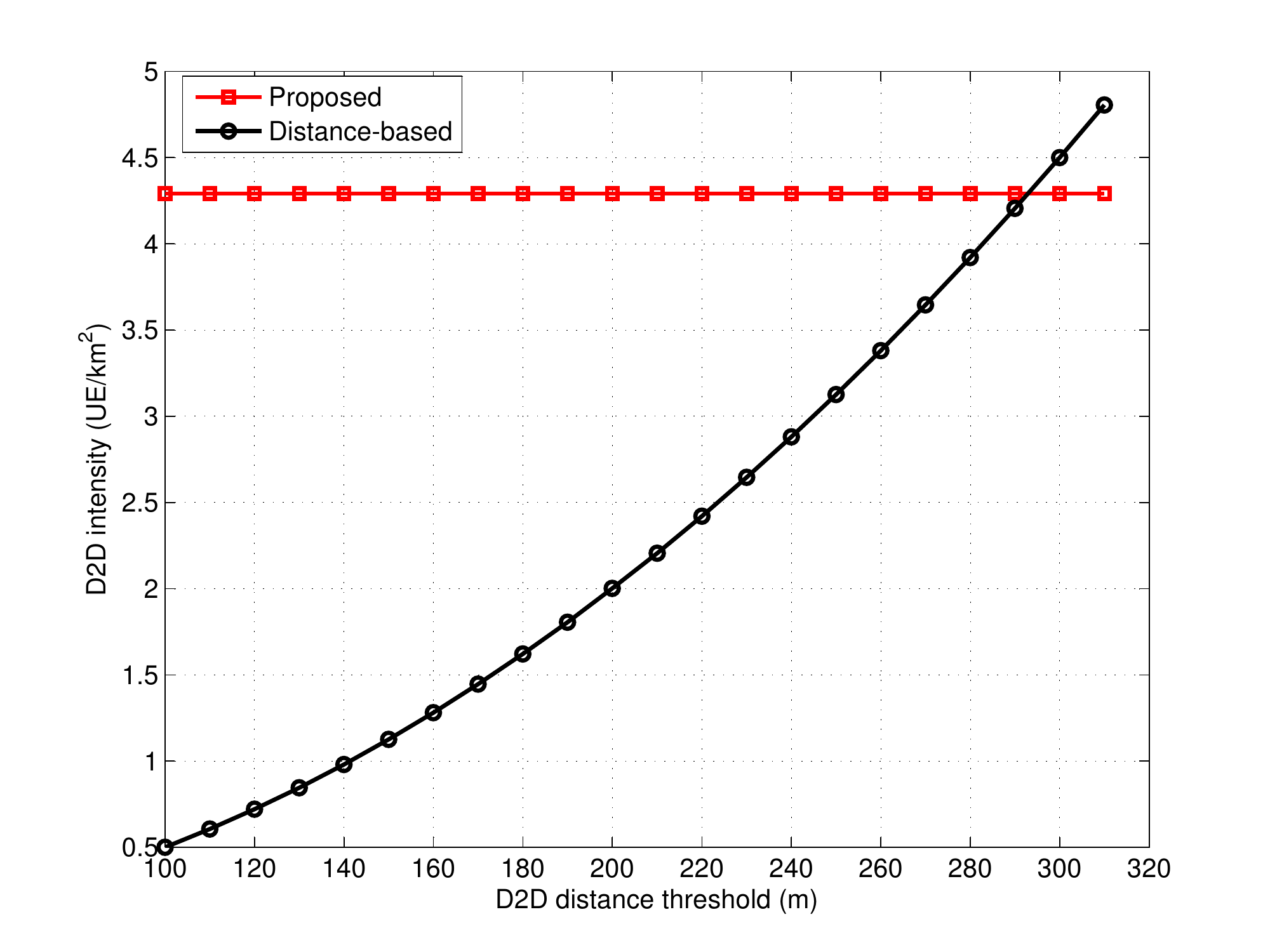}}
		\subfigure[]{\label{co3}\includegraphics[width=2.6 in]{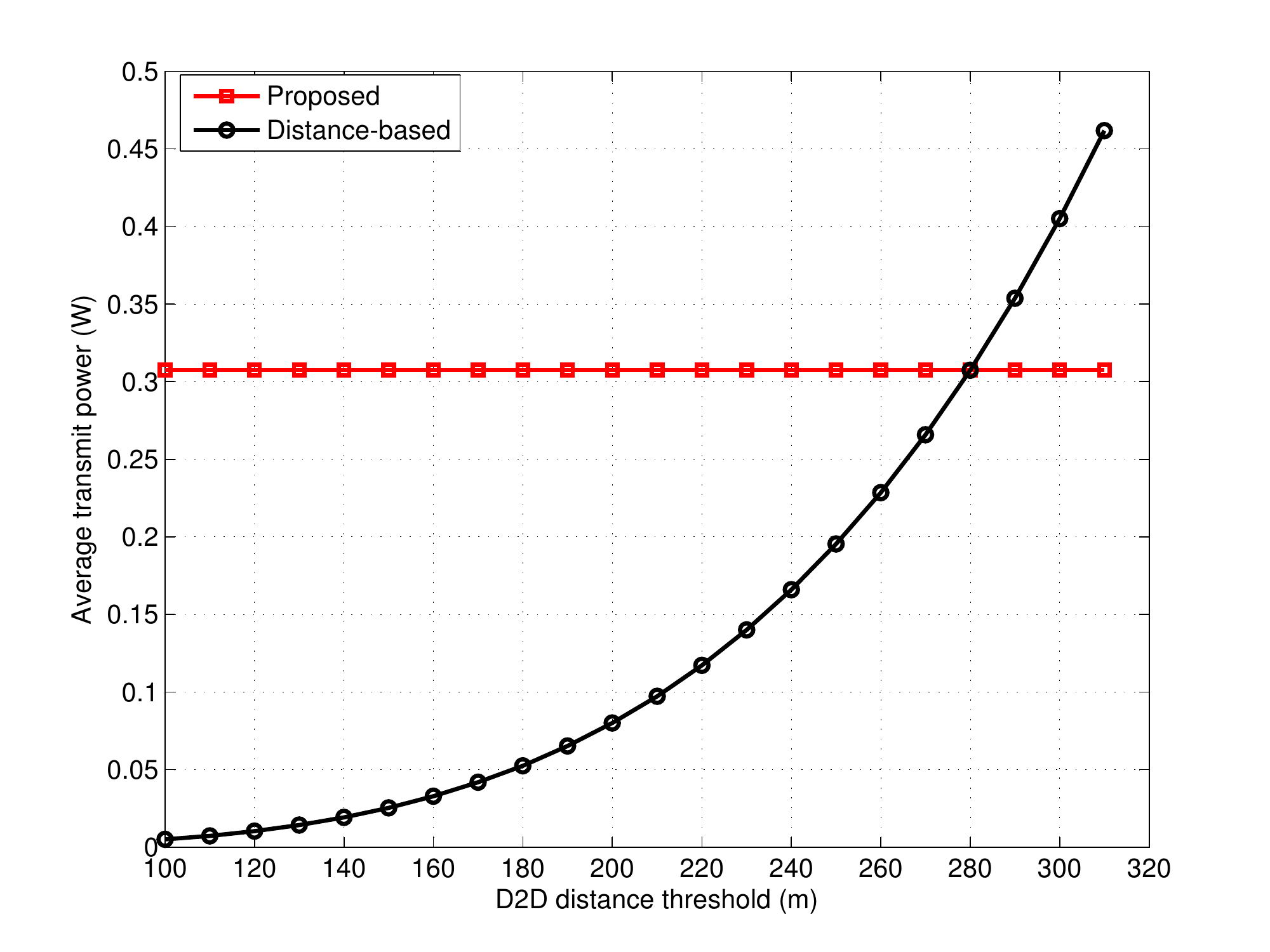}}
	\end{center}
	\caption{ Comparison between the proposed mode selection criterion and the distance-based mode selection criterion (for $T_d=1$, $\lambda=1$ BS/km$^2$, and $\rho_o=-70$ dBm) in terms of a) ${\rm SINR}$ outage probability, b) admitted D2D link intensity, and c) average transmit power.}
	\label{mvj}
\end{figure}

To compare the proposed mode selection scheme (which is based on the bias factor $T_d$) to the traditional mode selection scheme (which is based only on the D2D link distance), we plot Fig.~\ref{mvj}. The comparison is in terms of the ${\rm SINR}$ outage probability for cellular users (see Fig.~\ref{co1}), the intensity of admitted D2D links (see Fig.~\ref{co2}), and the average transmit power of D2D UEs (see Fig.~\ref{co3}). In order to conduct a fair comparison, we look at the point where both of the mode selection schemes result in the same intensity of admitted D2D UEs (distance threshold = 290 in Fig. \ref{co2}). Note that the different values of the admitted D2D UEs achieved by the distance-based mode selection scheme can be achieved via the proposed scheme by manipulating the bias factor $T_d$. Fig.~\ref{mvj}  shows that the proposed mode selection scheme outperforms the distance-based mode selection scheme. That is, for the same intensity of scheduled D2D links (see Fig.~\ref{co2}), the proposed mode selection scheme offers a lower ${\rm SINR}$ outage probability for cellular users (see Fig.~\ref{co1}) and lower average  transmit power of D2D UEs (see Fig.~\ref{co3}). 

The superiority of the proposed mode selection scheme is due to the following reasons. First, the proposed mode selection scheme correlates the D2D link distance with the cellular link distance, and hence, imposes interference protection around the cellular BSs. That is, for a potential UE, the D2D transmission mode is selected if and only if the power received from that UE at the nearest BS is less than $T_d \rho_o$. On the other hand, the 
distance-based mode selection scheme neglects the cellular link distance and schedules the D2D links according to a predefined D2D link distance threshold. That is, regardless of the cellular link distance, the D2D transmitter chooses the D2D mode if the D2D receiver is located at a distance less than a predefined threshold. Hence, there might be a D2D transmitter which is very close the a BS and selects the D2D mode, and thus, creates high cross-mode interference. For $T_d=1$, the proposed mode selection scheme  enforces all potential D2D UEs, which have a better cellular link quality, to communicate in the cellular mode. Hence, only the D2D UEs with relatively good D2D link quality communicate in the D2D mode which decreases the average transmit power of D2D UEs.

\subsection{Design Parameters}

In the following results, we study the effect two main design parameters on the performance of underlay D2D-enabled cellular networks, namely, the bias factor $T_d$ and the power control cutoff threshold $\rho_o$. $T_d$ is a crucial design parameter that controls the extent to which the D2D mode communication is enabled in the cellular network. As has been mentioned before, $T_d = 0$ completely disables the D2D communication and $T_d = \infty$ enforces all potential D2D UEs to communicate in D2D mode. 

We study the effect of $T_d$ in Figs. \ref{td1}-\ref{td4}. Fig. \ref{td1} shows that increasing $T_d$ increases the ${\rm SINR}$ outage probability. The reasons behind the degradation in ${\rm SINR}$ outage probability  with increasing  $T_d$ are as follows. First,  increasing $T_d$ increases the intensity of UEs selecting the D2D mode, and hence, increases the intensity of the interfering transmitters. Second, increasing $T_d$ decreases the interference protection around the cellular BSs, and hence, the  $\rm SINR$ outage probability of cellular users increases. Finally, for $T_d > 1$, the UEs spend more power to invert the channel towards the D2D receiver when compared to the power required to invert the channel towards the nearest BS in the cellular mode. Hence, $T_d > 1$ enforces a high transmit power for the D2D UEs, which deteriorates the overall ${\rm SINR}$ performance. Note that the cellular UEs experience a lower ${\rm SINR}$ outage probability (Fig. \ref{td1}). This is due to the interference protection around the BSs induced by the proposed mode selection scheme along with the channel inversion power control (see \textbf{Fact $\# 1$} in \textbf{Appendix~\ref{pth2}}), which is not the case for the D2D receivers.

\begin{figure}[th]
	\begin{center}
 \includegraphics[width=2.6 in]{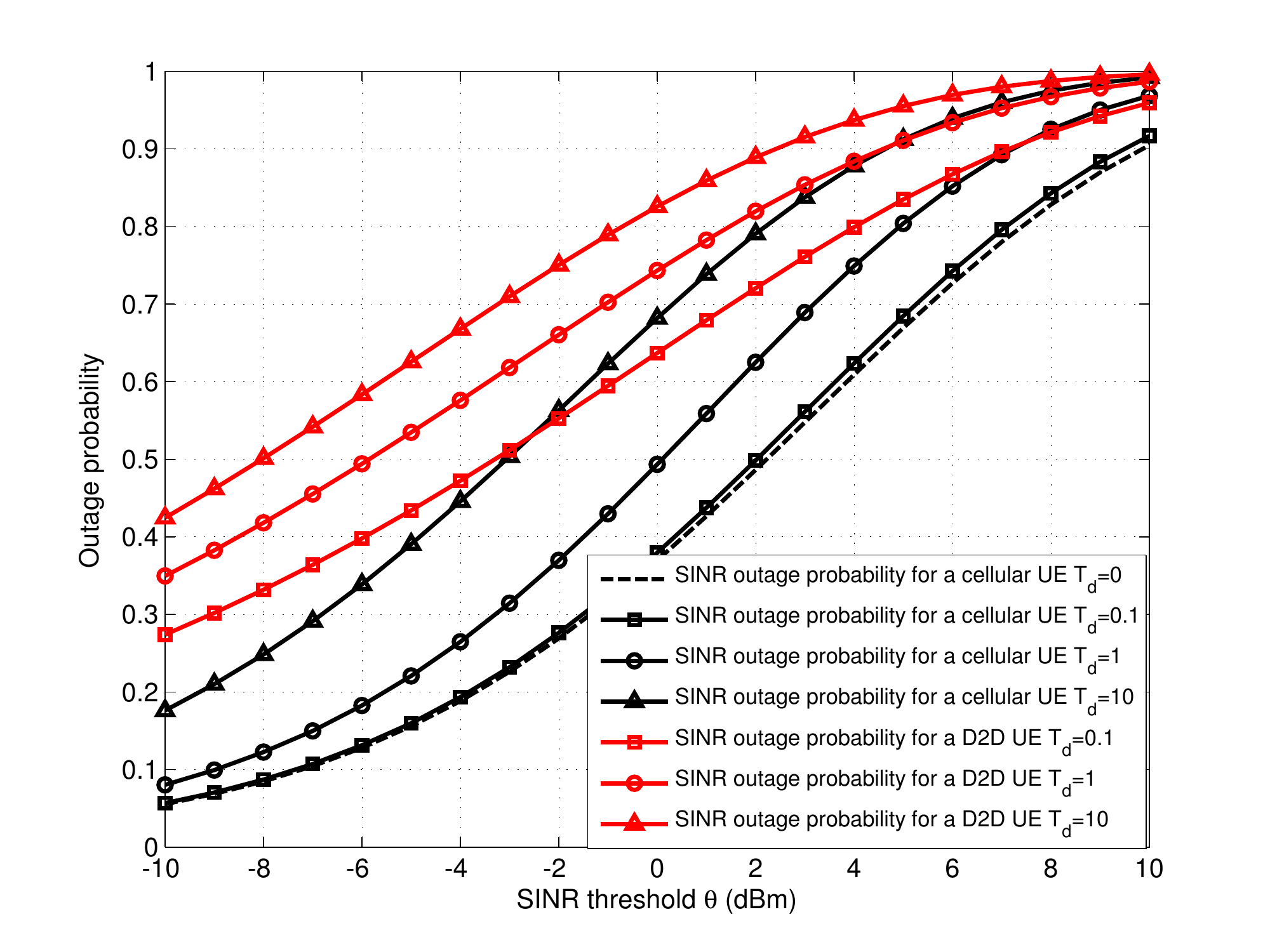}
	\end{center}
	\caption{ Effect of $T_d$ on ${\rm SINR}$ outage.}
	\label{td1}
\end{figure}

While Fig. \ref{td1} shows a negative impact of enabling D2D communication in cellular networks, Figs.~\ref{td2}-\ref{td4} show the positive impacts. For instance, Fig.~\ref{td2} shows that there exists an optimal $T_d$ that maximizes the rate ($\mathcal{R}$) for a generic potential D2D,  which is calculated as  
\begin{equation} 
\mathcal{R}=  \mathcal{P}_d \mathcal{R}_d + \frac{1}{2}\frac{(1-\mathcal{P}_d) \lambda}{(1-\mathcal{P}_d) p \mathcal{D}+(\mathcal{U}-p \mathcal{D}) (1-\mathcal{O}_p)} \mathcal{R}_c.
\label{minrate1}
\end{equation} 

At low $T_d$, most of the D2D UEs operate in the  cellular mode (i.e., $\mathcal{P}_d$ is very small) where the channel is shared by $\frac{\lambda}{(1-\mathcal{P}_d) p \mathcal{D}+(\mathcal{U}-p \mathcal{D}) (1-\mathcal{O}_p)}$ users on average. Note that we divide the cellular rate by 2 in equation (\ref{minrate1}) because the data requires at least two hops to reach its destination. Increasing $T_d$ increases $\mathcal{P}_d$ and more potential D2D UEs select the D2D mode where each UE exclusively uses the channel in the time domain\footnote{The channel is reused in the spatial domain by cellular and other D2D UEs.}, and hence, the average transmission rate  of a generic potential D2D UE increases. However, increasing $T_d$ beyond its optimal value degrades the average rate, despite the fact that each D2D UE exclusively uses the channel in the time domain, due to the high interference from the D2D network. This highly deteriorates the ${\rm SINR}$ performance.

\begin{figure}[th]
	\begin{center}
 \includegraphics[width=2.6 in]{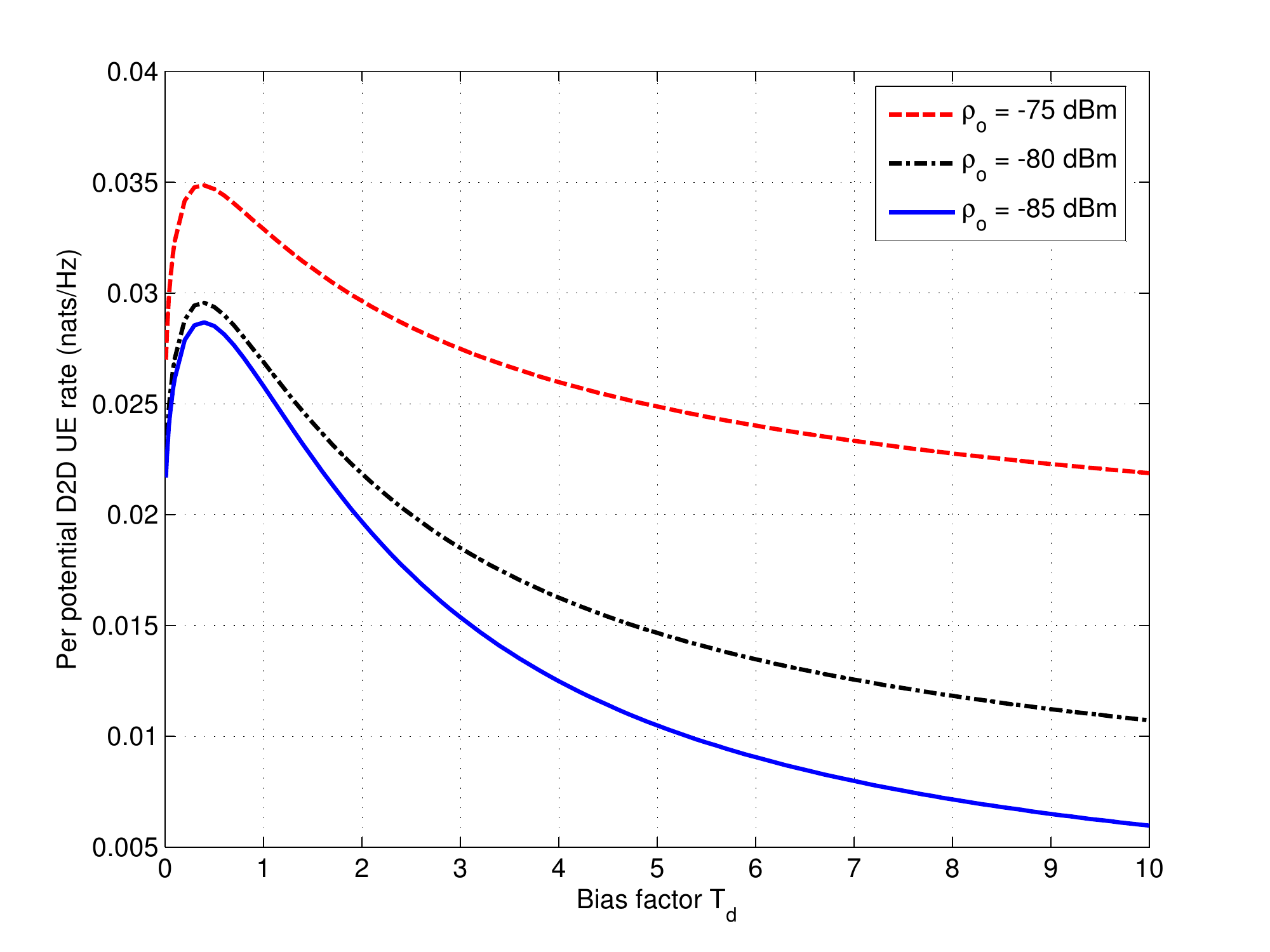}
	\end{center}
	\caption{ Effect of $T_d$ on the expected rate of a generic potential D2D UE.}
	\label{td2}
\end{figure}

Fig. \ref{td3} shows another figure of merit for underlay D2D communication in cellular networks. The figure  shows that with proper adjustment of bias threshold $T_d$, the total network capacity, which is calculated as $ \mathcal{T}= \mathcal{P}_d p\mathcal{D} \mathcal{R}_d + \lambda \mathcal{R}_c$, can be maximized. In particular,  D2D communication improves the spatial frequency reuse efficiency, and hence, increases the total network capacity. However, for high values (i.e., higher than its optimal value) of $T_d$, the total network capacity deteriorates as a result of the poor ${\rm SINR}$ performance due to the increased intensity of interfering D2D UEs and the decreased interference protection region around cellular BSs. 

\begin{figure}[th]
	\begin{center}
 \includegraphics[width=2.6 in]{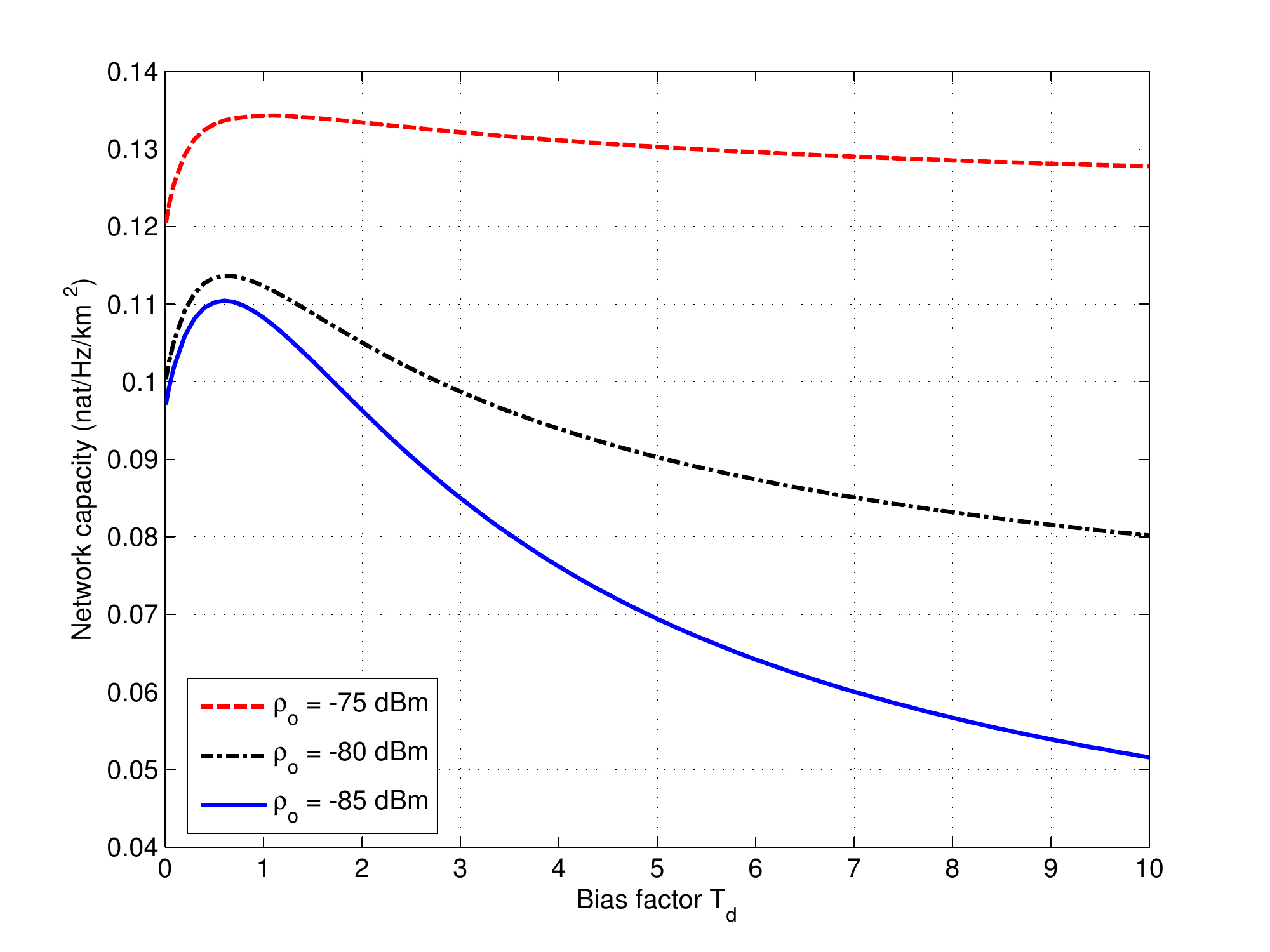}
	\end{center}
	\caption{ Effect of $T_d$ on the total network capacity $\mathcal{T}$.}
	\label{td3}
\end{figure}

Fig. \ref{td3} shows that enabling D2D communication can also be exploited to decrease the transmit powers of the potential D2D UEs. With $T_d = 1$, a potential D2D UE chooses the mode which costs less transmit power for channel inversion, and hence, $T_d=1$ is the optimal biasing factor that minimizes the transmit powers of the UEs. Note that the reduction in transmit power is more prominent for higher values of the cutoff threshold $\rho_o$ because of the higher transmit power required for channel inversion. 

\begin{figure}[th]
	\begin{center}
 \includegraphics[width=2.6 in]{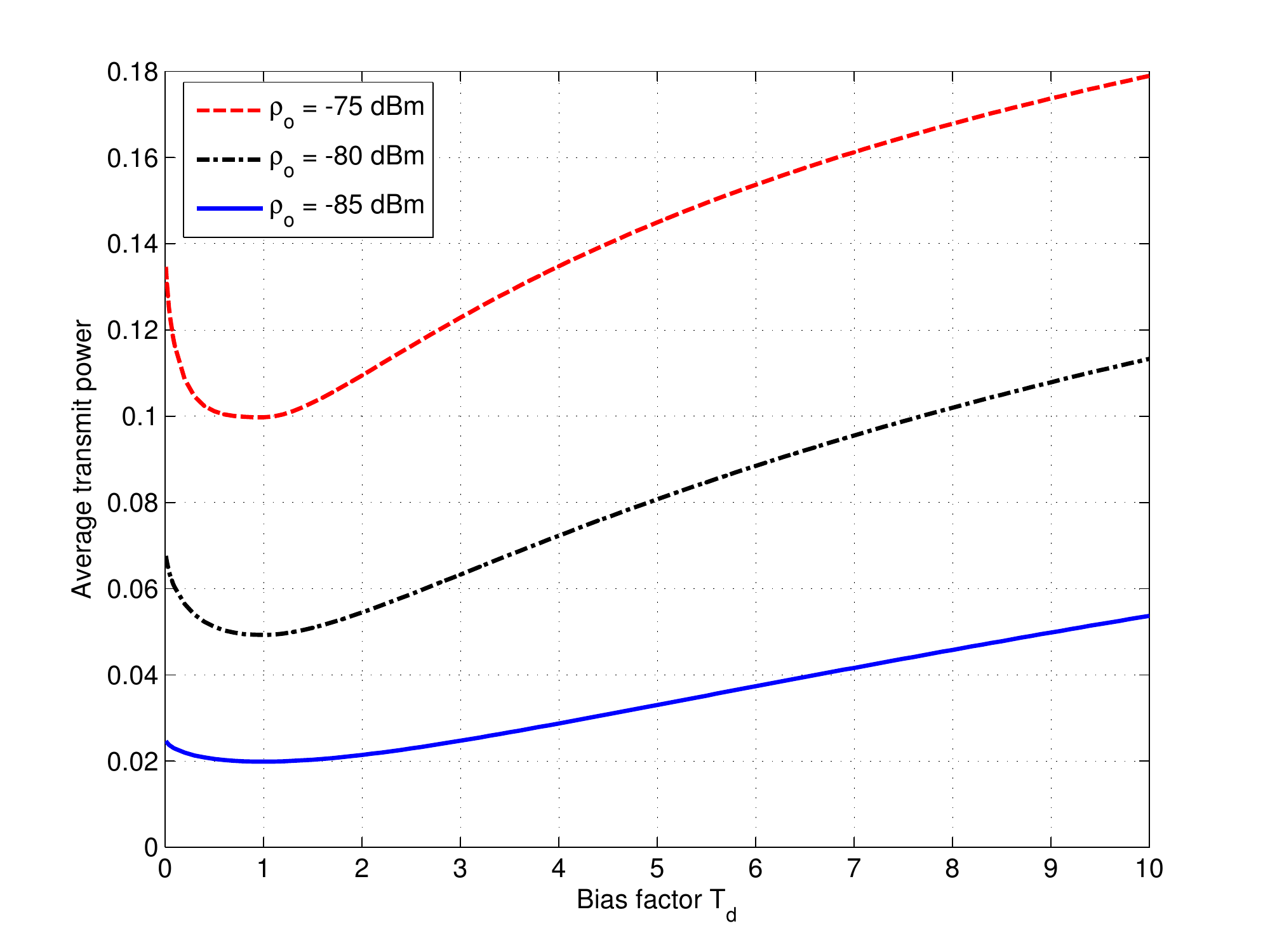}
	\end{center}
	\caption{Effect of $T_d$ on expected transmit power of a potential D2D UE.}
	\label{td4}
\end{figure}

Figs. \ref{valid}-\ref{td4} also reveal that the power control cutoff threshold is a very crucial design parameter that affects all of the performance metrics. On one hand, Figs. \ref{valid}-\ref{td3} show that increasing $\rho_o$ improves the network performance in terms of ${\rm SINR}$ outage as well as expected link capacity. On the other hand, Fig.~\ref{td4} shows that increasing $\rho_o$ increases the transmit power consumption of the UEs by increasing the average transmit power. In order to see the effect of cutoff threshold more clearly, we plot Fig.~\ref{rho}. 

\begin{figure}[th]
	\begin{center}
		\subfigure[]{\label{rho_d}\includegraphics[width=2.6 in]{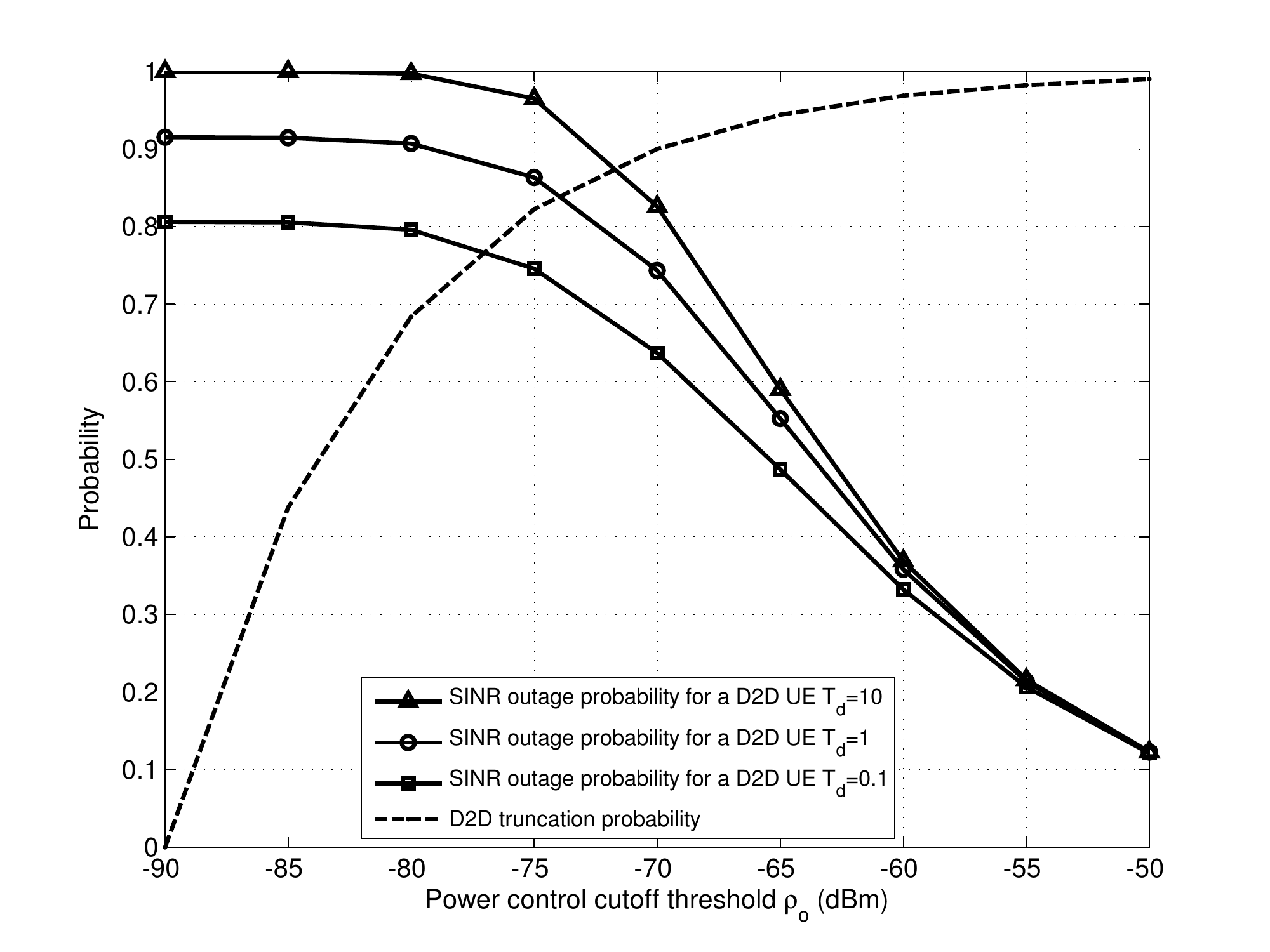}}
	  \subfigure[]{\label{rho_c}\includegraphics[width=2.6 in]{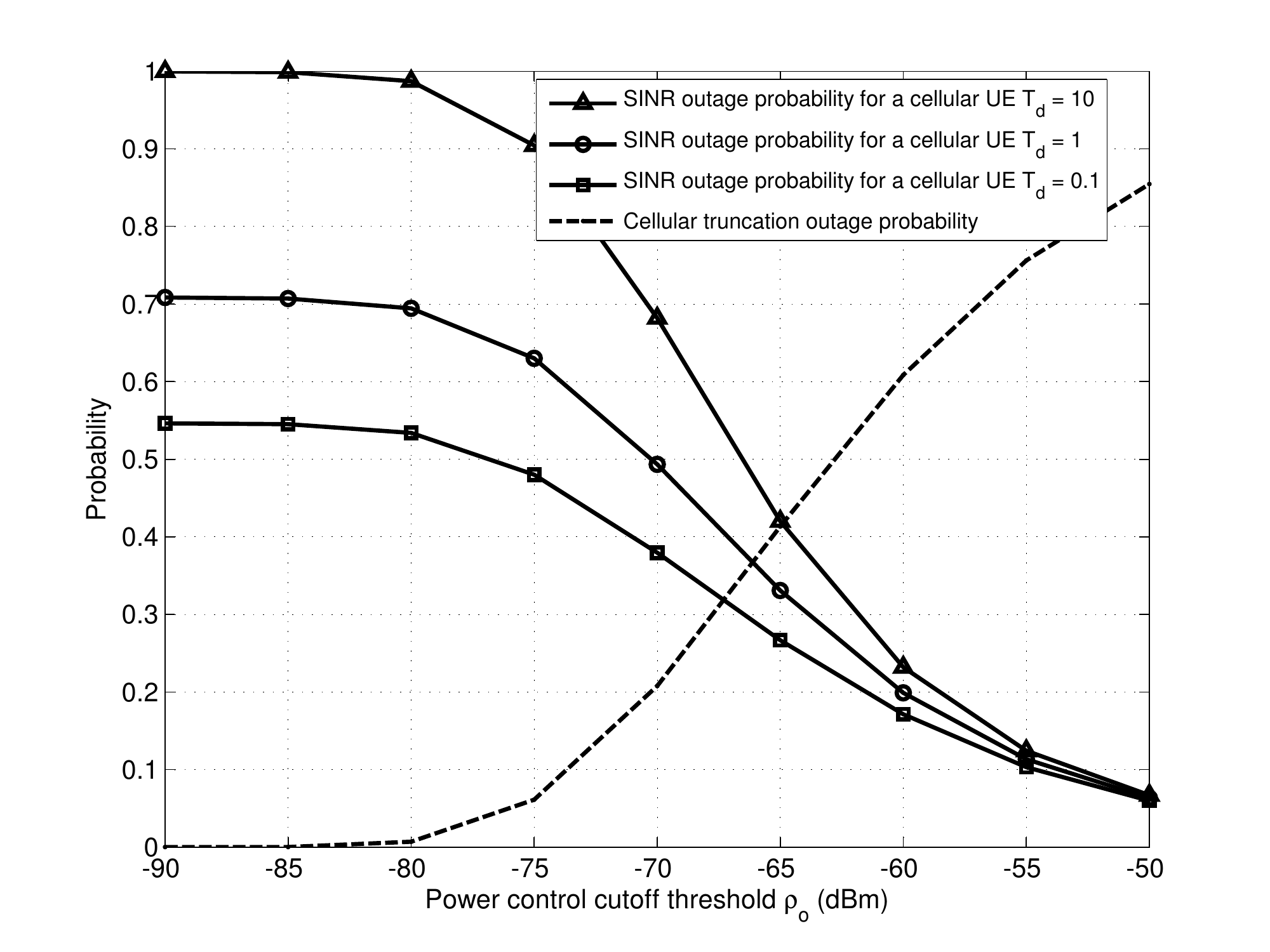}}
	\end{center}
	\caption{ Effect of power control cutoff threshold $\rho_o$ on  ${\rm SINR}$ outage and truncation outage at $T_d =1$ for a) D2D links and b) cellular links.}
	\label{rho}
\end{figure}

Both Figs.~\ref{rho_d} and \ref{rho_c} show that increasing the cutoff threshold decreases the ${\rm SINR}$ outage probability for both the D2D links and cellular links. However, increasing $\rho_o$ requires a higher transmit power to invert the channel between transmitters and receiver pairs which increases the truncation outage. That is, at low values of $\rho_o$, the ${\rm SINR}$ dominates the outage probability due to the low power of the useful signal. On the other hand, increasing the cutoff threshold $\rho_o$ increases the power of the useful signal and decreases the ${\rm SINR}$ outage probability at the expense of increased truncation outage probability.  As has been mentioned earlier, increasing $\rho_o$ increases the average transmit power per UE in both  cellular and D2D modes (as shown in Fig.~\ref{td4}), which increases the aggregate interference. However, the contribution of  $\rho_o$  to the useful signal power dominates the increased  aggregate interference power and results in an improved ${\rm SINR}$ performance. 

It is worth mentioning that the behavior of truncation outage differs in the D2D case from that in the cellular case due to the different distance distributions between the transmitter receiver pairs. In particular, increasing the cutoff threshold significantly limits the D2D communication due to the assumed uniform distribution of the D2D link distance. Note that for D2D applications in the context of social networks, where shorter D2D link distances are more likely to occur~\cite{dist1, dist2} (which is different from the assumed distribution for the D2D link distance), the effect of the cutoff threshold on the D2D truncation will be less prominent.

\subsection{Discussions} \label{dis}

The proposed analytical framework is general and captures the distance-based mode selection scheme as a special case. More specifically, by setting $T_d = \infty$ and manipulating the D2D link distance via the truncation outage  (i.e., by varying $\rho_o$), our model reduces to the distance-based mode selection scheme. The results show that setting a high value of $T_d$ for any value of $\rho_o$ results in a high degradation in $\rm SINR$ performance. This implies that considering the D2D link distance only as the mode selection criterion will not provide an efficient solution to the mode selection problem.
 On the other hand, the proposed mode selection scheme accounts for both the cellular link distance and D2D link distance and introduces a fine-tuned control for mode selection via the bias factor $T_d$. Based on the operator's objective, a suitable value of $T_d$ can be selected using the presented framework. Compared to the mode selection scheme based on the D2D link distance only, the proposed scheme results in a $25\%$ reduction in the ${\rm SINR}$ outage probability for cellular UEs and a $15 \%$ reduction in the average transmit power of the D2D UEs for the same intensity of scheduled D2D UEs.

The results show that the power control cutoff threshold $\rho_o$ introduces an important tradeoff for the network performance. On one hand, increasing $\rho_o$ improves the ${\rm SINR}$ performance by decreasing the ${\rm SINR}$ outage probability and increasing the link capacity. However, increasing $\rho_o$ requires a higher transmit power to satisfy the power control cutoff threshold and increases the truncation outage probability. One solution to this problem is to set $\rho_o$ to a relatively low value to ensure an acceptable truncation outage probability while implementing an interference management technique (e.g., cognition) to improve the ${\rm SINR}$ performance.

The analysis and results shown is this paper present a pessimistic evaluation for the D2D communication. That is, in addition to the PPP assumption for the cellular BSs which gives pessimistic bounds on the performance metrics \cite{trac, trac2, pp-cellular},  the one channel assumption along with the aggressive spectrum access for the UEs operating in the D2D mode introduce massive cross-mode and intra-mode interference. Furthermore, assuming that the D2D receiver is uniformly distributed around the D2D transmitter implies that larger D2D link distances are more likely to occur, which increases the transmit power of the D2D devices and the associated interference. Nevertheless, under these pessimistic assumptions, underlay D2D communication has shown performance improvement in terms of spatial frequency reuse, link capacity, and total network capacity, for the cellular networks. Furthermore, D2D communication can be exploited to decrease the cellular truncation outage probability which results due to the maximum limited power control. Therefore, we expect significant performance gains for the D2D underlay network with interference coordination (e.g., via cognition)\footnote{It was shown in \cite{our-globe, our-tmc} that interference coordination via cognition highly decreases the outage probability in multi-tier cellular networks.} and assuming different D2D link distance distribution (e.g., in the context of social networking). Note that the analysis in this paper can be extended to capture interference coordination via cognition among the D2D UEs using the same methodology as presented in \cite{our-tmc}. 

The results show that enforcing all potential D2D UEs to communicate in the D2D mode results in a significant degradation in   network performance. This is because, wih  $T_d > 1$, the interference signal power  received at a BS from an individual D2D transmitter is greater than the intended uplink signal power  $\rho_0$ (i.e., the interference power is $T_d \rho_0$). Hence, the performance gain in terms of the spatial frequency reuse efficiency is offset by the performance degradation in terms of the ${\rm SINR}$ performance, and consequently, the overall network performance deteriorates.

Finally, we emphasize that although the framework is developed for single-tier cellular networks, extension of the framework to multi-tier networks is systematic and straightforward. In fact, as shown in \cite{our-uplink}, if all network tiers have the same cutoff threshold and path-loss exponent, the multi-tier network analysis reduces to the single-tier network case and the developed framework captures this special case. The model can be extended to more general cases with different cutoff thresholds and path-loss exponents, for the different network tiers, via the same methodology as in \cite{our-uplink}.

\section{Conclusion} \label{con}

 We have proposed a biasing-based mode selection method for D2D-enabled cellular networks. The extent to which D2D communication is enabled and the amount of traffic offloaded to the D2D communication mode is controlled via the bias value $T_d$. We have developed an analytical paradigm to evaluate outage and rate in the proposed D2D-enabled cellular network. The results have shown that underlay D2D communication is capable of improving the system performance in terms of spatial frequency reuse, link capacity, and total network capacity. Significant performance gains in the total network performance can be expected from underlay D2D communication with the deployment of  interference management methods.  The analysis has revealed that enforcing all potential D2D UEs to communicate in the D2D mode, or selecting the D2D mode according to the D2D link distance only results in a degraded network performance when compared  to the proposed mode selection scheme.

The developed analytical framework has revealed two important design parameters for D2D-enabled cellular networks, namely, the bias factor and the power control cutoff threshold. In particular, the bias factor controls the interference protection for cellular users and the intensity of enabled D2D communication. The results have shown that there exists an optimal D2D bias factor that depends on the network objective. The power control cutoff threshold controls the tradeoff between the ${\rm SINR}$ outage and truncation outage. An appropriate value of the cutoff threshold that balances between the two outage probabilities can be obtained by using the developed framework. The presented framework can be extended to include interference management techniques in order to improve the performance of D2D-enabled cellular network.




\appendices

\section{Proof of Lemma~\ref{lem1}} \label{proof1}

Since the intended receiver is uniformly distributed in the communication range (proximity) of the potential D2D transmitter, the {\em pdf} of the D2D link distance ($r_d$) is given by $f_{r_d}(r) = \frac{2r}{R^2}$ $0 \leq r \leq R$. According to the PPP assumption of the locations of the BSs, the cellular link distance $r_c$ is Rayleigh distributed with {\em pdf} $f_{r_c}(r)=2 \pi \lambda r e^{-\lambda \pi r^2}$  $0 \leq r$ \cite{trac}. The probability of selecting the D2D mode can be expressed as

\small
\begin{align}
\mathbb{P}&\left\{T_d r_d^{-\eta_d} > r_c^{-\eta_c} \right\} = \mathbb{P}\left\{r_c > \frac{r_d^{\frac{\eta_d}{\eta_c}}}{T_d^{\frac{1}{\eta_c}}} \right\} \notag \\
&= \int_{0}^{R} \frac{2r}{R^2} e^{-\pi \lambda \frac{r^{\frac{2\eta_d}{\eta_c}}}{T_d^{\frac{2}{\eta_c}}}} dr \notag \\
&\stackrel{(i)}{=}  \frac{ \eta_c T_d^{\frac{2}{\eta_d}}}{ \eta_d R^2} \left(\frac{1}{\pi \lambda}\right)^{\frac{\eta_c}{\eta_d}} \gamma\left(\frac{\eta_c}{\eta_d},\pi \lambda \left(\frac{R^{\eta_d}}{T_d}\right)^{\frac{2}{\eta_c}} \right) \notag
\end{align}
\normalsize
where $(i)$ follows from changing variables $\left(x=\pi \lambda \frac{r^{\frac{2\eta_d}{\eta_c}}}{T_d^{\frac{2}{\eta_c}}}\right)$ and integration.

\section{Proof of Lemma~\ref{lem2}} \label{proof2}

Let $X_c = \rho_o r_c^{\eta_c}$ and $X_d = \rho_o r_c^{\eta_d}$ denote the unconditional transmit powers required to invert the channel towards the nearest BS and D2D receiver, respectively. Then, from the {\em pdf} of $r_c$ and $r_d$ given in \textbf{Appendix \ref{proof1}}, the {\em pdf} of $X_c$ can be obtained as $f_{X_c}(x) = \frac{2 \pi \lambda x^{\frac{2}{\eta_c}-1} e^{-\pi \lambda \left(\frac{x}{\rho_o}\right)^{\frac{2}{\eta_c}}}}{\eta_c  \rho_o^{\frac{2}{\eta_c}}}$, $0 \leq x \leq \infty$,  and  the {\em pdf} of $X_d$ can be obtained as  $f_{X_d}(x) = \frac{2 x^{\frac{2}{\eta_d}-1}}{\eta_d (\rho_oR)^\frac{2}{\eta_d}} $, $0 \leq x \leq \rho_o R^{\eta_d}$. According to the mode selection scheme we have $P_d = \left\{X_d : X_d \leq T_d X_c \right\} $ and due to the maximum transmit power constraint, we have $P_d \leq P_u$. Then, given that a user is operating in the D2D mode,  the {\em pdf} of its transmit power  is given by

\small
\begin{align}
f_{P_d}(x) &=  \int_x^{\infty} \frac{f_{X_d|T_dX_c}(x|y) \mathbb{P}\left\{X_d \leq y\right\} f_{T_d X_c}(y)}{\mathbb{P}\left\{X_d < T_dX_c\right\}} dy \notag \\
 &= \frac{1}{\mathbb{P}\left\{X_d \leq T_d X_c \right\}} \int_{\frac{x}{T_d}}^{\infty} \frac{2 x^{\frac{2}{\eta_d}-1}}{\eta_d \min\left( y T_d,P_u\right)^\frac{2}{\eta_d}} \notag \\
&\frac{\min\left(T_d y,P_u\right)^{\frac{2}{\eta_d}}}{P_u^\frac{2}{\eta_d}} \frac{2 \pi \lambda y^{\frac{2}{\eta_c}-1} e^{-\pi \lambda \left(\frac{y}{\rho_o}\right)^{\frac{2}{\eta_c}}}}{\eta_c  \rho_o^{\frac{2}{\eta_c}}} dy \notag \\
 &=\frac{2 x^{\frac{2}{\eta_d}-1}}{\eta_d P_u^\frac{2}{\eta_d} \mathbb{P}\left\{X_d \leq T_d X_c \right\}}\left[  \int_{\frac{x}{T_d}}^{\frac{P_u}{T_d}}  \frac{2 \pi \lambda y^{\frac{2}{\eta_c}-1} }{ \eta_c  \rho_o^{\frac{2}{\eta_c}} } \right. \notag \\
& e^{-\pi \lambda \left(\frac{y}{\rho_o}\right)^{\frac{2}{\eta_c}}} dy + \left.\int_{\frac{P_u}{T_d}}^{\infty}  \frac{2 \pi \lambda y^{\frac{2}{\eta_c}-1} e^{-\pi \lambda \left(\frac{y}{\rho_o}\right)^{\frac{2}{\eta_c}}}}{\eta_c  \rho_o^{\frac{2}{\eta_c}} } dy\right] \notag \\
&= \frac{2 x^{\frac{2}{\eta_d}-1}  (\pi \lambda)^{\frac{\eta_c}{\eta_d}} e^{-\pi \lambda \left(\frac{x}{T_d \rho_o}\right)^\frac{2}{\eta_c}} }{\eta_c (T_d \rho_o)^{\frac{2}{\eta_d}} \gamma\left(\frac{\eta_c}{\eta_d},\pi \lambda \left(\frac{P_u}{T_d\rho_o}\right)^{\frac{2}{\eta_c}}\right)}, 0 \leq x \leq P_u.  
\end{align}
\normalsize

The $\alpha$th moment of $P_d$ is obtained as $\int_{0}^{P_u} x^\alpha f_{P_d}(x) dx$.


\section{Proof of Lemma~\ref{lem4}} \label{proof4}


Let $\tilde{X}_c = \rho_o \tilde{r}_c^{\eta_c}$, then the {\em pdf} of $\tilde{X}_c$ can be obtained as $f_{\tilde{X}_c}(x) = \frac{2 \pi \lambda x^{\frac{2}{\eta_c}-1} e^{-\pi \lambda \left(\frac{x}{\rho_o}\right)^{\frac{2}{\eta_c}}}}{\eta_c  \rho_o^{\frac{2}{\eta_c}}\left(1-e^{-\pi \lambda \left(\frac{P_u}{\rho_o}\right)^{\frac{2}{\eta_c}}}\right)}$, $0 \leq x \leq P_u$. According to the mode selection scheme we have $\tilde{P}_4 = \left\{\tilde{X}_c : \tilde{X}_c \leq \frac{X_d}{T_d}\right\}$, and the {\em pdf } of $\tilde{P}_4$ can be obtained as 

\small
\begin{align}
f_{\tilde{P}_4}(x)  &= \int_x^{\frac{P_u}{T_d}} \frac{f_{\tilde{X}_c|\frac{X_d}{T_d}}(x|y) \mathbb{P}\left\{\tilde{X}_c \leq y\right\} f_{\frac{X_d}{T_d}}(y) dy}{\mathbb{P}\left\{\tilde{X}_c < \frac{X_d}{T_d}\right\}} \notag \\ 
&= \int_x^{\frac{{P}_u}{T_d}} \frac{2 \pi \lambda x^{\frac{2}{\eta_c}-1} e^{-\pi \lambda \left(\frac{x}{\rho_o}\right)^{\frac{2}{\eta_c}}}}{\eta_c  \rho_o^{\frac{2}{\eta_c}} \left(1-e^{-\pi \lambda \left(\frac{P_u}{\rho_o}\right)^{\frac{2}{\eta_c}}}\right)} 
\frac{2 y^{\frac{2}{\eta_d}-1}}{\eta_d \left(\frac{P_u}{T_d}\right)^\frac{2}{\eta_d}}  dy \notag \\
&= \frac{2 \pi \lambda x^{\frac{2}{\eta_c}-1} \left({P}_u^\frac{2}{\eta_d} - (T_dx)^\frac{2}{\eta_d}\right) e^{-\pi \lambda \left(\frac{x}{\rho_o}\right)^{\frac{2}{\eta_c}}}}{\mathbb{P}\left\{\tilde{X}_c < \frac{X_d}{T_d}\right\} \eta_c  \rho_o^{\frac{2}{\eta_c}} P_u^\frac{2}{\eta_d} \left(1-e^{-\pi \lambda \left(\frac{P_u}{\rho_o}\right)^{\frac{2}{\eta_c}}}\right)}, \notag \\
& \quad \quad \quad \quad \quad \quad \quad \quad 0 \leq x \leq \frac{P_u}{\max(T_d,1)}.  
\end{align}
\normalsize
Note that both the random variables $\tilde{X}_c$ and $X_d$ have the finite support domain $[0,P_u]$; therefore, the value of $T_d$ will affect the probability $\mathbb{P}\left\{\tilde{X}_c < \frac{X_d}{T_d}\right\}$. For $T_d \leq 1$, we have

\small
\begin{align}
\mathbb{P}\left\{\tilde{X}_c < \frac{X_d}{T_d}\right\}&=  \int_0^{{T_d P_u}} \frac{1 - e^{-\pi \lambda \left(\frac{x}{T_d\rho_o}\right)^{\frac{2}{\eta_c}}}}{\left(1-e^{-\pi \lambda \left(\frac{P_u}{\rho_o}\right)^{\frac{2}{\eta_c}}}\right)} \frac{2 x^{\frac{2}{\eta_d}-1}}{\eta_d P_u^\frac{2}{\eta_d}} dx  \notag \\
&+ \mathbb{P}\left\{T_d P_u \leq X_d \leq P_u\right\}  \notag \\
&=1 - T_d^\frac{2}{\eta_d} +  \frac{T_d^\frac{2}{\eta_d}}{{\left(1-e^{-\pi \lambda \left(\frac{P_u}{\rho_o}\right)^{\frac{2}{\eta_c}}}\right)}} \notag \\
&-\frac{\eta_c(T_d \rho_o)^{\frac{2}{\eta_d}} \gamma\left(\frac{\eta_c}{\eta_d},\pi \lambda \left(\frac{P_u}{\rho_o}\right)^{\frac{2}{\eta_c}}\right)}{\eta_d P_u^\frac{2}{\eta_d } (\pi  \lambda)^{\frac{\eta_c}{\eta_d}} {\left(1-e^{-\pi \lambda \left(\frac{P_u}{\rho_o}\right)^{\frac{2}{\eta_c}}}\right)} }. 
\label{pp1}
\end{align}

For $T_d > 1$, we have

\small
\begin{align}
\mathbb{P}\left\{\tilde{X}_c < \frac{X_d}{T_d}\right\}&=  \int_0^{P_u} \frac{1 - e^{-\pi \lambda \left(\frac{x}{T_d\rho_o}\right)^{\frac{2}{\eta_c}}}}{\left(1-e^{-\pi \lambda \left(\frac{P_u}{\rho_o}\right)^{\frac{2}{\eta_c}}}\right)} \frac{2 x^{\frac{2}{\eta_d}-1}}{\eta_d P_u^\frac{2}{\eta_d}} dx  \notag \\
&=\frac{1 }{\left(1-e^{-\pi \lambda \left(\frac{P_u}{\rho_o}\right)^{\frac{2}{\eta_c}}}\right)} \notag \\ 
&-\frac{\eta_c(T_d \rho_o)^{\frac{2}{\eta_d}} \gamma\left(\frac{\eta_c}{\eta_d},\pi \lambda \left(\frac{P_u}{T_d\rho_o}\right)^{\frac{2}{\eta_c}}\right)}{\eta_d P_u^\frac{2}{\eta_d } (\pi  \lambda)^{\frac{\eta_c}{\eta_d}} {\left(1-e^{-\pi \lambda \left(\frac{P_u}{\rho_o}\right)^{\frac{2}{\eta_c}}}\right)} }.
\label{ppp2}
\end{align}
Then we combine (\ref{pp1}) and (\ref{ppp2}) in a compact form to arrive at (\ref{pp3}). Note that $\mathbb{P}\left\{\tilde{X}_c < \frac{X_d}{T_d}\right\} = \mathbb{P}\left\{\tilde{r}_c < \frac{r_d}{T_d}\right\}$. The moments of $\tilde{P}_4$ are evaluated by $\int_0^{P_u} x^\alpha f_{\tilde{P}_4}(x) dx $.

\section{Proof of Theorem~\ref{th1}} \label{pth1}

The theorem is proved by obtaining the LTs of $\mathcal{I}_{dd}$ and $\mathcal{I}_{cd}$ and then substituting back in (\ref{SINR1}).  An approximate LT of the aggregate interference on a D2D receiver located at the origin from other D2D transmitters is obtained based on the assumption that the interfering D2D UEs constitute a PPP rather than a Poisson hole process. The approximate LT is obtained as

\small
\begin{align} 
 &\mathcal{L}_{\mathcal{I}_{dd}}(s) = \mathbb{E}\left[e^{- \underset{u_i \in \tilde{\mathbf{\Phi}}_d}{\sum} \tilde{P}_{d_i} h_i \left\|u_i\right\|^{-\eta_d}}\right]\notag \\
 &\stackrel{(ii)}{=}  \mathbb{E}_{ \tilde{\mathbf{\Phi}}_d}\left[\underset{u_i \in \tilde{\mathbf{\Phi}}_d}{\prod} \mathbb{E}_{P_d,h}\left[ e^{-s \tilde{P}_{d_i} h_i \left\|u_i\right\|^{-\eta_d}}\right]\right]\notag \\
 &\stackrel{(iii)}{=}   \exp\left( -  \frac{2 \pi \tilde{\mathcal{U}}_d}{\left|\mathbf{S}\right|} \int_{0}^{\infty}\mathbb{E}_{P_d,h}\left[  \left(1- e^{-s \tilde{P}_{d} h x^{-\eta_d}}\right)\right] xdx \right)\notag \\
&=\exp\left( - \frac{\pi \tilde{\mathcal{U}}_d}{\left|\mathbf{S}\right|} s^{\frac{2}{\eta_d}} \mathbb{E}\left[\tilde{P}_d^{\frac{2}{\eta_d}}\right] \Gamma\left(1+\frac{2}{\eta_d}\right) \Gamma\left(1-\frac{2}{\eta_d}\right)\right) 
\label{LT1}
\end{align}
\normalsize
where $\tilde{\Phi}_d$ is the PPP representing the locations of D2D transmitter, $\mathbb{E}_{x,y}[.]$ denotes the expectation with respect to the random variables $x$ and $y$, $(ii)$ follows from the independence between $\tilde{\mathbf{\Phi}}_d$, $P_d$, and $h$, and $(iii)$ is obtained by exploiting the moment generating functional of the PPP \cite{martin-book} and the fact that there in no interference protection around the D2D receivers. 

For the interference at the test D2D receiver from cellular UEs, an approximate LT of the aggregate interference on a D2D receiver located at the origin from other D2D transmitters is obtained based on the assumption that the interfering cellular UEs constitute a PPP with intensity $\lambda$ rather than a softcore process and that their transmit powers are independent. The intensity $\lambda$ is a consequence of scheduling only one user per BS to avoid intra-cell interference. The approximate LT is obtained as

\small
\begin{align} 
&\mathcal{L}_{\mathcal{I}_{cd}}(s) = \mathbb{E}\left[e^{- \underset{u_i \in \tilde{\mathbf{\Phi}}_c}{\sum} \tilde{P}_{c_i} h_i \left\|u_i\right\|^{-\eta_d}}\right]\notag \\
 &=  \mathbb{E}_{ \tilde{\mathbf{\Phi}}_c}\left[\underset{u_i \in \tilde{\mathbf{\Phi}}_c}{\prod} \mathbb{E}_{P_c,h}\left[ e^{-s \tilde{P}_{c_i} h_i \left\|u_i\right\|^{-\eta_d}}\right]\right]\notag \\
 &=   \exp\left( - 2 \pi  \lambda \int_{0}^{\infty}\mathbb{E}_{P_c,h}\left[  \left(1- e^{-s \tilde{P}_{c} h x^{-\eta_d}}\right)\right] xdx \right)\notag \\
&=\exp\left( - { \pi  \lambda s^\frac{2}{\eta_d}} \mathbb{E}\left[\tilde{P}_c^{\frac{2}{\eta_d}}\right]  \Gamma\left(1+\frac{2}{\eta_d}\right)  \Gamma\left(1-\frac{2}{\eta_d}\right)\right) 
\label{LT2}
\end{align}
\normalsize
 where $\tilde{\Phi}_c$ is the PPP representing the locations of the interfering UEs in the uplink. Note that we use $\eta_d$ as the path-loss exponent between any two UEs.  Based on Slivnyak's theorem for PPPs ~\cite{martin-book}, the obtained LTs are valid for any D2D receiver located in a generic position.

\section{Proof of Theorem~\ref{th2}} \label{pth2}

The theorem is proved by obtaining the LTs  of $\mathcal{I}_{dc}$ and $\mathcal{I}_{cc}$ for a BS located at the origin. 
 First, we derive an approximate LT of the interference from the cellular UEs. Note that orthogonal channel assignment per BS brings correlations among the locations of interfering UEs which highly complicate the analysis. Therefore, the derivation here is based on the following two facts and one key assumption.

\textbf{Fact $\#1$}: The UEs associate to the nearest BSs and all UEs employ truncated channel inversion protocol. Therefore, the average interference received from any single interfering UE is strictly less than $\rho_o$. 

\textbf{Fact $\#2$}: Each BS assigns a unique channel to each UE, hence, the intensity of the interfering UEs is $\lambda$. 

\textbf{Key assumption}: The interfering UEs  constitute a PPP and their transmit powers are independent. 

\small
\begin{align}
 &\mathcal{L}_{\mathcal{I}_{cc}}(s) = \mathbb{E}\left[e^{- \underset{u_i \in \tilde{\mathbf{\Phi}}_c}{\sum} P_{c_i} h_i \left\|u_i\right\|^{-\eta_c}}\right]\notag \\
&=  \mathbb{E}\left[\underset{u_i \in \tilde{\mathbf{\Phi}}_c}{\prod} e^{-s P_{c_i} h_i \left\|u_i\right\|^{-\eta_c}}\right]\notag \\
 &=  \exp\left( - 2 \pi   \lambda\int_{\left(\frac{P_c}{\rho_o}\right)^\frac{1}{\eta_c}}^{\infty}\mathbb{E}_{P_c}\left[\mathbb{E}_{h}\left[  \left(1- e^{-s P_{c} h x^{-\eta_c}}\right)\right]\right] xdx \right)\notag \\
&=\exp\left( - {2 \pi   \lambda} {s^\frac{2}{\eta_c}} \mathbb{E}_{P_c}\left[ P_c^\frac{2}{\eta_c}\right] \int_{(s \rho_o)^\frac{-1}{\eta_c}}^{\infty} \frac{x}{x^{\eta_c}+1 }dx \right)\notag \\
&\stackrel{(\eta_c=4)}{=}\exp\left( - { \pi  \lambda}{s^\frac{2}{\eta_c}} \mathbb{E}_{P_c}\left[ P_c^\frac{2}{\eta_c}\right]  \arctan(\sqrt{s \rho_o })\right).\notag 
\end{align}
\normalsize

 In the following, we obtain an approximate LT of the aggregate interference on a BS located at the origin from other D2D transmitters. This is based on the assumption that the interfering D2D UEs constitute a PPP rather than a Poisson hole process. Note that the interference protection imposed by the mode selection scheme is captured by insuring that the interference from any individual D2D UE is upper bounded by $T_d \rho_o$. The approximate LT is obtained as follows:

\small
\begin{align}
&\mathcal{L}_{\mathcal{I}_{cd}}(s) = \mathbb{E}\left[e^{-s\mathbbm{1}(P_d \left\|x_i\right\|^{-\eta_c} < T_d \rho_o ) \sum_{x_i \in \Phi} P_d h \left\|x_i\right\|^{-\eta_c}}\right] \notag \\
&= \mathbb{E}\left[e^{-s\mathbbm{1}( \left\|x_i\right\| > \left(\frac{P_d }{T_d  \rho_o}\right)^\frac{1}{\eta_c} ) \sum_{x_i \in \Phi} P_d h \left\|x_i\right\|^{-\eta_c}}\right] \notag \\
 &= \mathbb{E}\left[\underset{x_i \in \Phi_d}{\prod} e^{-s\mathbbm{1}( \left\|x_i\right\| > \left(\frac{P_d }{ T_d \rho_o}\right)^\frac{1}{\eta_c} ) P_d h \left\|x_i\right\|^{-\eta_c}}\right] \notag \\
 &=  \exp\left( - \frac{2 \pi \tilde{\mathcal{U}}_d}{\left|\textbf{S}\right|}  \int_{\left(\frac{P_d}{T_d  \rho_o}\right)^\frac{1}{\eta_c}}^{\infty} \mathbb{E}_{P_d}\left[\mathbb{E}_{h}\left[  \left(1- e^{-s P_{d} h x^{-\eta_c}}\right)\right]\right] xdx \right)\notag \\
 &=  \exp\left( - \frac{2 \pi \tilde{\mathcal{U}}_d}{\left|\textbf{S}\right|} s^\frac{2}{\eta_c} \mathbb{E}_{P_d}\left[ P_{d}^\frac{2}{\eta_c}\right] \int_{\left(\frac{1}{sT_d \rho_o}\right)^\frac{1}{\eta_c}}^{\infty} \left( \frac{y }{y^{\eta_c}+1  }\right) dy \right)\notag \\
&\stackrel{(\eta_c=4)}{=}  \exp\left( - \frac{ \pi \tilde{\mathcal{U}}_d}{\left|\textbf{S}\right|} s^\frac{2}{\eta_c} \mathbb{E}_{P_d}\left[ P_{d}^\frac{2}{\eta_c}\right] \arctan\left(\sqrt{sT_d \rho_o}\right)\right). \notag \\
\end{align}
\normalsize

Based on Slivnyak's theorem for PPPs ~\cite{martin-book}, the obtained LTs are valid for any BS located in a generic position.

\end{document}